\documentclass[letterpaper]{vldb}

\usepackage[disable]{todonotes}

\usepackage{graphicx}
\usepackage{balance}
\usepackage{bm}
\usepackage{color}
\usepackage{listings}
\usepackage{colortbl}
\usepackage{verbatim} 
\usepackage{textcomp}
\usepackage{afterpage}

\usepackage{amsmath}
\usepackage{amssymb}

\usepackage{amsthm}

\usepackage{algorithm}
\usepackage{algorithmicx}
\usepackage{algpseudocode}

\usepackage{tikz}

\usepackage{array}
\usepackage{tabularx}
\usepackage{multirow}

\usetikzlibrary{shapes,backgrounds,arrows}
\usepackage[textfont=bf,labelfont=bf,subrefformat=parens,position=top]{subfig}
\usepackage{listings}
\usepackage{caption}

\usepackage{paralist}

\usepackage{url}
\usepackage{breakurl}
\usepackage[breaklinks]{hyperref}
\usepackage{textcomp}

\newcommand{\partitle}[1]{\smallskip\noindent\textbf{#1}.}

\DeclareRobustCommand{\BG}[1]{{\todo[color=red!40,inline]{\textbf{Boris says:}{#1}}}}

\newtheorem{theo}{Theorem}
\newtheorem{defi}{Definition}
\newtheorem{lem}{Lemma}

\newtheorem{exam}{Example}

\newenvironment{theoremproof}[2][\bf Theorem]{\begin{trivlist}
\item[\hskip \labelsep {#1}\hskip \labelsep {#2}] \em}{\end{trivlist}}

\definecolor{black}{rgb}{0,0,0}
\definecolor{grey}{rgb}{0.8,0.8,0.8}
\definecolor{red}{rgb}{1,0,0}
\definecolor{green}{rgb}{0,1,0}
\definecolor{darkgreen}{rgb}{0,0.5,0}
\definecolor{darkpurple}{rgb}{0.5,0,0.5}
\definecolor{darkdarkpurple}{rgb}{0.3,0,0.3}
\definecolor{blue}{rgb}{0,0,1}
\definecolor{shadegreen}{rgb}{0.95,1,0.95}
\definecolor{shadeblue}{rgb}{0.95,0.95,1}
\definecolor{shadered}{rgb}{1,0.85,0.85}
\definecolor{shadegrey}{rgb}{0.85,0.85,0.85}
\definecolor{oddRowGrey}{rgb}{0.80,0.80,0.80}
\definecolor{evenRowGrey}{rgb}{0.85,0.85,0.85}

\newcommand{\defas}{:=}

\newcommand{\mathtext}[1]{\thickspace\text{#1}\thickspace}
\newcommand{\mathtab}{\thickspace\thickspace\thickspace}
\newcommand{\projection}{\Pi}
\newcommand{\selection}{\sigma}

\newcommand{\union}{\cup}

\newcommand{\join}{\bowtie}

\newcommand{\asingleton}[2]{\{#1 \to {#2}\}}

\newcommand{\schema}[1]{\textsc{Sch}(#1)}

\newcommand{\RAPlus}{{\cal RA}^+}

\newcommand{\upAttr}{{\cal U}}

\newcommand{\xid}{T}
\newcommand{\xidDomain}{\mathbb{T}}
\newcommand{\version}{\nu}
\newcommand{\versionDomain}{\mathbb{V}}
\newcommand{\tid}{id}
\newcommand{\tidOf}{\tid}
\newcommand{\tidDomain}{\mathbb{I}}
\newcommand{\start}[1]{Start({#1})}
\newcommand{\finish}[1]{End(#1)}
\newcommand{\history}{H}

\newcommand{\up}{u}

\newcommand{\relUpdate}[5]{{\cal U}[#1,#2,#3,#4](#5)}
\newcommand{\relInsert}[4]{{\cal I}[#1,#2,#3](#4)}
\newcommand{\relDelete}[4]{{\cal D}[#1,#2,#3](#4)}
\newcommand{\commitOp}[3]{{\cal C}[#1,#2](#3)}
\newcommand{\doCommit}[3]{\textsc{com}[#1,#2](k)}
\newcommand{\db}{D}
\newcommand{\upMarker}[4]{#1_{#2,#3}^{#4}}
\newcommand{\cMarker}[3]{C_{#1,#2}^{#3}}

\newcommand{\upMarkers}{\mathbb{A}}
\newcommand{\upMark}{{\cal A}}

\newcommand{\rel}{R}
\newcommand{\relV}[3]{{#1}[#2,#3]}
\newcommand{\relCV}[2]{{#1}[#2]}
\newcommand{\relVE}[3]{{#1}_{ext}[#2,#3]}
\newcommand{\validIn}{\textsc{validIn}}
\newcommand{\validEx}{\textsc{validEx}}

\newcommand{\dbV}[2]{{\db}[#1,#2]}
\newcommand{\dbCV}[1]{{\db}[#1]}

\newcommand{\genAnnotOp}{\alpha}
\newcommand{\annotOp}[4]{\genAnnotOp_{{#1},{#2},{#3}}}

\newcommand{\vFilt}[1]{{\gamma}_{#1}}

\newcommand{\vMerge}{{\mu}}
\newcommand{\versionOf}{versionOf}
\newcommand{\idOf}{idOf}
\newcommand{\isMax}{isMax}
\newcommand{\isStrictMax}{isStrictMax}

\newcommand{\validAt}{\textsc{validAt}}
\newcommand{\hasUp}{\textsc{updated}}

\newcommand{\ppSR}{\mathbb{N}[X]}
\newcommand{\ppSRV}{\mathbb{N}[X]^{\version}}

\newcommand{\ract}{\mathbb{R}}

\newcommand{\semK}{{\cal K}}
\newcommand{\mvK}{{\cal K}^{\version}}
\newcommand{\mvOf}[1]{{#1}^{\version}}
\newcommand{\mvDom}{K^{\version}}

\newcommand{\NXv}{\mathbb{N}[X]^{\version}}
\newcommand{\congr}[1]{[{#1}]_\sim}

\newcommand{\numInSum}[1]{n(#1)}
\newcommand{\nthOfK}[2]{#1[#2]}

\newcommand{\liftH}[1]{{#1}^{\version}}

\newcommand{\thead}[1]{{\cellcolor{black}{\textcolor{white}{\textbf{#1}}}}}

    \makeatletter
     \newbox\sf@box
       {\def\sf@one{#1}%
        \def\sf@two{#2}%
        \setbox\sf@box\hbox
          \bgroup}%
       {  \egroup
        \ifx\@empty\sf@two\@empty\relax
          \def\sf@two{\@empty}
        \fi
        \ifx\@empty\sf@one\@empty\relax
          \subfloat[\sf@two]{\box\sf@box}%
        \else
          \subfloat[\sf@one][\sf@two]{\box\sf@box}%
        \fi}
     \makeatother

\newcommand{\eat}[1]{}
\newcommand{\hasCreated}{\textsc{hasCreated}}
\newcommand{\impred}{\textsc{immPred}}
\newcommand{\prooftitle}[1]{\smallskip\noindent\underline{#1}:}
\newcommand{\mathbigtab}{\mathtab\mathtab\mathtab}

\begin{document}

\lstdefinestyle{psql}
{
tabsize=2,
basicstyle=\small\upshape\ttfamily,
language=SQL,
morekeywords={PROVENANCE,BASERELATION,INFLUENCE,COPY,ON,TRANSPROV,TRANSSQL,TRANSXML,CONTRIBUTION,COMPLETE,TRANSITIVE,NONTRANSITIVE,EXPLAIN,SQLTEXT,GRAPH,IS,ANNOT,THIS,XSLT,MAPPROV,cxpath,OF,TRANSACTION,SERIALIZABLE,COMMITTED,INSERT,INTO,WITH,SCN,UPDATED},
extendedchars=false,
keywordstyle=\color{blue},
mathescape=true,
escapechar=@,
sensitive=true
}

\lstdefinestyle{rsl}
{
tabsize=3,
basicstyle=\small\upshape\ttfamily,
language=C,
morekeywords={RULE,LET,CONDITION,RETURN,AND,FOR,INTO,REWRITE,MATCH,WHERE},
extendedchars=false,
keywordstyle=\color{blue},
mathescape=true,
escapechar=@,
sensitive=true
}

\lstdefinestyle{pseudocode}
{
  tabsize=3,
  basicstyle=\small,
  language=c,
  morekeywords={if,else,foreach,case,return,in,or},
  extendedchars=true,
  mathescape=true,
  literate={:=}{{$\gets$}}1 {<=}{{$\leq$}}1 {!=}{{$\neq$}}1 {append}{{$\listconcat$}}1 {calP}{{$\cal P$}}{2},
  keywordstyle=\color{blue},
  escapechar=&,
  numbers=left,
  numberstyle={\color{green}\small\bf}, 
  stepnumber=1, 
  numbersep=5pt,
}

\lstdefinestyle{xmlstyle}
{
  tabsize=3,
  basicstyle=\small,
  language=xml,
  extendedchars=true,
  mathescape=true,
  escapechar=£,
  tagstyle={\color{blue}},
  usekeywordsintag=true,
  morekeywords={alias,name,id},
  keywordstyle={\color{red}}
}


\lstset{style=psql}

\title{Reenactment for Read-Committed Snapshot Isolation}
\subtitle{(Long Version)}

\author{
\alignauthor
 \hspace{-1cm} Bahareh Sadat Arab{\large $\,^{\bigstar}$}, Dieter Gawlick{\large $\,^{\blacklozenge}$}, Vasudha Krishnaswamy{\large $\,^{\blacklozenge}$},\\ Venkatesh Radhakrishnan{\large $\,^{\square}$}, Boris Glavic{\large $\,^{\bigstar}$}\\
 \affaddr{ \hspace{-2cm} {\large $\,^{\bigstar}$}Illinois Institute of Technology \hspace{4cm}  {\large $\,^{\blacklozenge}$}Oracle \hspace{3cm} {\large $\,^{\square}$}LinkedIn}\\
 \affaddr{\normalsize \hspace{-3mm} \{barab@hawk.,bglavic@\}iit.edu, \{dieter.gawlick,vasudha.krishnaswamy\}@oracle.com, vradhakrishnan@linkedin.com}\\
}
\maketitle

\begin{abstract}
Provenance for transactional updates is critical for many applications such as auditing and debugging of transactions. Recently, we have introduced \emph{MV-semirings}, an extension of the semiring provenance model that supports updates and transactions. Furthermore, we have proposed \emph{reenactment}, a declarative form of replay with provenance capture, as an efficient and non-invasive method for computing this type of provenance. However, this approach is limited to the snapshot isolation (SI) concurrency control protocol while many real world applications apply the read committed version of snapshot isolation (RC-SI) to improve 
performance at the cost of consistency. We present non-trivial extensions of the model and reenactment approach to be able to compute provenance of RC-SI transactions efficiently. In addition, we develop techniques for applying reenactment across multiple RC-SI transactions. Our experiments demonstrate that our implementation in the GProM system 
supports efficient re-construction and querying of provenance.

\end{abstract}





\section{Introduction}
\label{sec:introduction}


Tracking the derivation of data through a history of transactional updates, i.e., tracking the provenance of such operations, is critical for many applications including auditing, data integration, probabilistic databases, and post-mortem debugging of transactions. For example, by exposing data dependencies, provenance provides proof of how data was derived, by which operations, and at what time. 
Until recently, no solution did exist for tracking the provenance of updates run as part of concurrent transactions. 
%

\partitle{MV-semirings and Reenactment}
In previous work~\cite{AG16}, we have introduced \textbf{MV\--sem\-i\-rings} (multi-version sem\-i\-rings). MV-semirings extend the sem\-i\-ring provenance framework~\cite{KG12} with support for transactional updates. We have introduced a low-overhead implementation of this model in our GProM system~\cite{AG14} using a novel declarative replay technique (\textbf{re\-en\-act\-ment}~\cite{AG16}). This finally makes this type of provenance available to applications using the \textit{snapshot isolation} (\textbf{SI}) concurrency control protocol. Figure~\ref{fig:Reenactment} illustrates how reenactment is applied to retroactively compute provenance for updates and transactions based on replay with provenance capture. Consider the database states induced by a history of concurrently executed transactions.  
With our approach, a user can request the provenance of any transaction executed in the past, e.g., Transaction $T_2$ in the example. Using reenactment, a temporal query is generated that simulates the transaction's operations within the context of the transactional history and this query is instrumented for provenance capture. This so-called reenactment query is guaranteed to return the same results (updated versions of the relations modified by the transaction) as the original transaction. In the result of the reenactment query, each tuple is annotated with its complete derivation history: 1) from which previous tuple versions was it derived and 2) which updates of the transaction affected it. Importantly, reenactment only requires an audit log (a log of SQL commands executed in the past) and time travel (query access to the transaction time history of tables) to function. That is, no modifications to the underlying database system or transactional workload are required. Many DBMS including Oracle~\cite{WEB:ORACLE-FBA}, 
DB2, and MSSQL~\cite{LL09} support a query-able audit log and time travel. If a system does not natively  support this functionality we can implement it using extensibility mechanisms (e.g., triggers). 
Snapshot isolation is a widely applied protocol (e.g., supported by Oracle, PostgreSQL, MSSQL, and many others). However, the practical applicability of reenactment is limited by the fact that many real world applications use statement-level snapshots instead of transaction-level snapshots. Using statement-level snapshots improves performance and timeliness of data even though this comes at the cost of reduced consistency.
In this work, we present the non-trivial extensions that are necessary to support 
statement-level snapshot isolation (isolation level \lstinline!READ COMMITTED! in the aforementioned systems).

%
%

\begin{figure*}[t]
  \centering
  \begin{minipage}{0.45\linewidth}
    \begin{minipage}{1\linewidth}
        \centering
  \includegraphics[width=6cm]{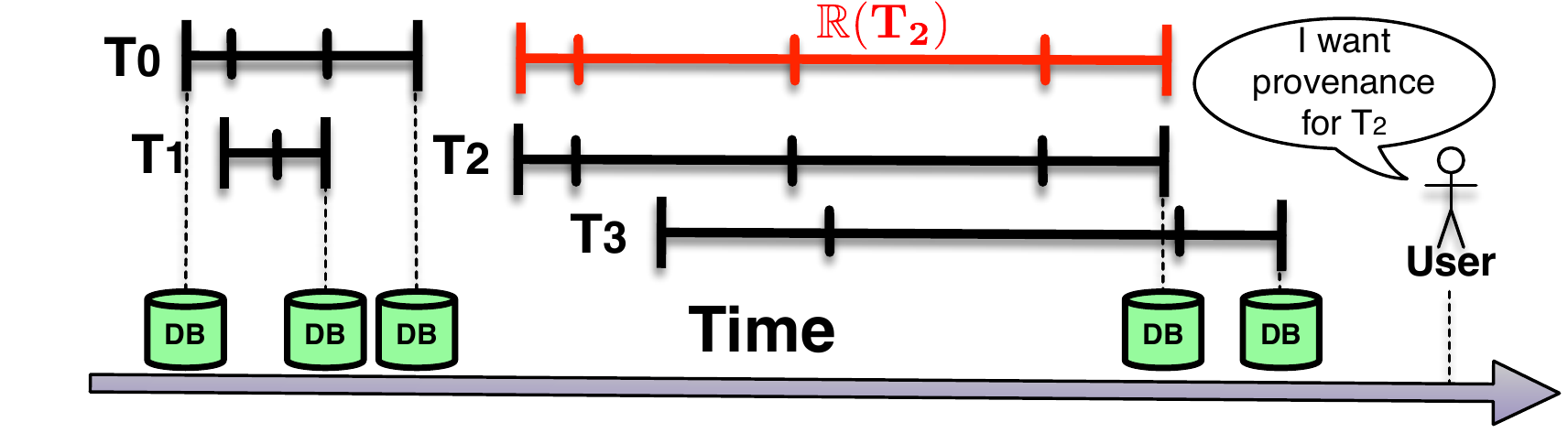}
   \caption{Reenactment}
   \label{fig:Reenactment}
    \end{minipage}
    \begin{minipage}{1\linewidth}
\resizebox{0.9\textwidth}{!}{
  \begin{minipage}{1.1\textwidth}
   \centering {\large \bf Employee}\\[2mm]
    \begin{tabular}{l|c|c|c|l}
     & \thead{ID} & \thead{Name} & \thead{Position} & \\ \cline{2-4}
    $\cMarker{T_0}{6}{1}(\upMarker{I}{T_0}{2}{1}(x_1))$ & 101 & Mark Smith & Software Engineer & $e_1$ \\
    $\cMarker{T_0}{6}{2}(\upMarker{I}{T_0}{3}{2}(x_2))$ & 102 & Susan Sommers & Software Architect & $e_2$ \\     
    $\cMarker{T_0}{6}{3}(\upMarker{I}{T_0}{4}{3}(x_3))$ &  103 & David Spears & Test Assurance & $e_3$ \\ \cline{2-4}
    \end{tabular}\\[2mm]

    \centering {\large \bf Bonus}\\[2mm]
    \begin{tabular}{l|c|c|c|l}
     & \thead{ID} & \thead{EmpID} &  \thead{Amount}&  \\ \cline{2-4}
	$\cMarker{T_1}{10}{4}(\upMarker{I}{T_1}{8}{4}(x_4))$  & 1 & 101 & 1000 &  $b_1$ \\
	$\cMarker{T_2}{14}{5}(\upMarker{I}{T_2}{12}{5}(x_5))$  & 2 & 102 & 2000 &  $b_2$ \\
	$\cMarker{T_4}{18}{6}(\upMarker{I}{T_4}{16}{6}(x_6))$  & 3 & 103 & 500 &  $b_3$  \\ \cline{2-4} 
     \end{tabular}
     
  \end{minipage}
}\\[2mm]
  \caption{Running example database instance}
  \label{fig:running-example-instance}
\end{minipage}

\end{minipage}
\begin{minipage}{0.54\linewidth}
  \begin{minipage}{1\linewidth}
  \centering
\resizebox{1\columnwidth}{!}{
  \begin{minipage}{1.3\columnwidth}
   \centering {\large \bf Employee}\\[2mm]
    \begin{tabular}{l|c|c|c|l}
     & \thead{ID} & \thead{Name} & \thead{Position} & \\ \cline{2-4}
   $\cMarker{T_{7}}{26}{1}(\upMarker{U}{T_{7}}{21}{1}(\cMarker{T_0}{6}{1}(\upMarker{I}{T_0}{2}{1}(x_1)))))$ & 101 & Mark Smith & \cellcolor{shadered} Software Architecture & ${e_1}'$ \\
      \cline{2-4}
    \end{tabular}\\[2mm]

    \centering {\large \bf Bonus}\\[2mm]
    \begin{tabular}{l|c|c|c|l}
     & \thead{ID} & \thead{EmpID} &  \thead{Amount}&  \\ \cline{2-4}
	$\cMarker{T_{7}}{26}{4}(\upMarker{U}{T_{7}}{22}{4}(\cMarker{T_1}{10}{4}(\upMarker{I}{T_1}{8}{4}(x_4))))$ & 1 & 101 & \cellcolor{shadered} 2000 &  ${b_1}'$ \\
	$\cMarker{T_8}{24}{7}(\upMarker{I}{T_8}{23}{7}(\cMarker{T_0}{6}{1}(\upMarker{I}{T_0}{2}{1}(x_1))))$  &\cellcolor{shadered}4 & \cellcolor{shadered}101 & \cellcolor{shadered}500 &  $b_4$  \\\cline{2-4} 
     \end{tabular}
     
  \end{minipage}
}\\[2mm]
\caption{New and modified tuples after execution of the example history (version 26). Modified attribute values and new tuples are shown with shaded background.}
  \label{fig:updated-example-instance}
\end{minipage}
\begin{minipage}{1\linewidth}
    \centering
  \resizebox{1\columnwidth}{!}{
  \begin{minipage}{1.65\columnwidth}
 \begin{tabular}{l|ccc||ccc||c||c|}
& \multicolumn{3}{c||}{\bf Bonus} & \multicolumn{3}{c||}{\bf
  Bonus Provenance} & {\bf $u_1$} &
  {\bf $u_2$} \\[1mm]
& \thead{ID} & \thead{EmpID} & \thead{Amount} 
& \thead{P(B,ID)} & \thead{P(B,EmpID)} & \thead{P(B,Amount)}
& \thead{$\upAttr_1$ }
& \thead{$\upAttr_2$ } 
\\[1mm]   
$\cMarker{T_{7}}{26}{4}(\upMarker{U}{T_{7}}{22}{4}(b_1))$ & 1 & 101 & 2000 & 1 & 101 & 1000 & F & T \\ \cline{2-9} 
\end{tabular}\\[2mm]
  \end{minipage}
}
  \caption{Relational encoding of provenance restricted to Transaction $\xid_{7}$. Only tuples modified by this transaction are included and the derivation history of tuples is limited to updates of this transaction. Variables are encoded as actual tuples.}
  \label{fig:rel-enc-example}
\end{minipage}
\end{minipage}
\end{figure*}

\partitle{Snapshot Isolation (SI)} Snapshot isolation~\cite{BB95} is a widely applied multi-versioning concurrency control protocol. Under SI each trans\-ac\-tion $\xid$ sees a private
snapshot of the database containing changes of transactions that have committed
before $\xid$ started and $\xid$'s own changes. SI disallows concurrent transactions to update the same data item. This is typically implemented by using write locks where transactions waiting for a lock have to abort if the transaction currently holding the lock commits.


\partitle{Read Committed Snapshot Isolation (RC-SI)} Under RC-SI each statement of a transaction sees changes of transactions that committed before the
statement was executed. In this
paper we assume the RC-SI semantic as implemented by Oracle, i.e., a statement waiting for a write-lock is restarted once the transaction holding the lock commits. This guarantees that each statement sees a consistent snapshot of the database. 
\begin{figure}[t]
  \centering


\resizebox{1\columnwidth}{!}{  
   \begin{tabular}{|l|l|c|}
     \hline
     \multicolumn{1}{c|}{\thead{T}} & \multicolumn{1}{c|}{\thead{SQL}}
     & \thead{Time} \\ \hline
     \rowcolor{shadeblue} $T_{7}$ & \lstinline! UPDATE Employee  SET Position='Software Architect'! &20\\
     \rowcolor{shadeblue} &\lstinline! WHERE ID=101; ! & \\  [1mm] 
     $T_{7}$ & \lstinline! UPDATE Bonus  SET Amount = Amount + 1000! &21\\
                                    &\lstinline! WHERE ID=101; ! & \\  [1mm]                                
     \rowcolor{shadeblue}$T_{8}$ & \lstinline! INSERT INTO Bonus (EmpID, Amount) ! & 22 \\ 
     \rowcolor{shadeblue}& \lstinline!  (SELECT ID, 500 FROM Employee ! & \\
     \rowcolor{shadeblue}& \lstinline!  WHERE Position='Software Engineer'); ! & \\[1mm]
                   
     $T_{8}$ & \lstinline! COMMIT;! &  23\\ [1mm]
    \rowcolor{shadeblue}  $T_{7}$ &\lstinline! SELECT Amount INTO amounts FROM Bonus  ! & 24 \\
     \rowcolor{shadeblue} 		 &\lstinline! WHERE ID=101; ! & \\  [1mm]                                
    $T_{7}$ &\lstinline! COMMIT;! & 25\\ \hline
   \end{tabular}
}
\caption{Example Transactional History}
\label{fig:Transactions Example}
\end{figure}
\begin{exam}
\label{ex:running-example}
Consider the example database shown in
Figure~\ref{fig:running-example-instance} storing information about employees and the bonuses they received. Ignore the annotations to the left of each tuple for now.
Two transactions have been executed concurrently (Figure~\ref{fig:Transactions Example}) using the \textit{RC-SI} protocol. 
In this example, all software engineers got a bonus of \$1000 while software architects received \$2000.  
Suppose administrator Bob executed transaction $\xid_{7}$ to update the position of Mark Smith to reflect his recent promotion to architect and update his bonus accordingly (increasing it by \$1000).
Concurrently, user Alice executed Transaction $T_8$ to implement the company's new policy of giving an additional bonus of \$500 to all \emph{software engineers}.
All new and updated tuples for relations \texttt{Employee} and \texttt{Bonus} after the execution of these transactions are shown in Figure~\ref{fig:updated-example-instance} (updated attributes are marked in red). Bob has executed a query at the end of his transaction ($\xid_{7}$) to double check the bonus amount for Mark expecting a single bonus of \$2000 instead of the actual result (a second bonus of \$500). The unexpected second bonus is produced by Transaction $\xid_8$, because this transaction did not see the uncommitted change of $\xid_7$ reflecting Mark's promotion. Thus, Mark was considered to still be a software engineer and received the corresponding \$500 bonus. This kind of error is hard to debug, because it only materializes if the execution of the two transactions is interleaved in a certain way and would not occur in any serializable schedule.
\end{exam}

By exposing data dependencies among tuple versions (e.g., the \$500 bonus for Mark is based on the previous version of Mark's tuple $e_1$ in the \texttt{Employee} table) and by recording which operations created a tuple version (e.g., the updated \$2000 bonus for Mark was produced by the second update of $\xid_7$), the MV-semiring provenance model greatly simplifies debugging of transactions. We now give an overview of our model and then present our extensions for RC-SI. 

\section{The MV-Semiring Model}
\label{sec:prov-model-vers}

Using MV-semirings (multi-version semirings), provenance is represented as annotations on tuples, i.e., each tuple is annotated with its derivation history (provenance). 

\partitle{$\semK$-relations}
We briefly review  the semiring provenance framework~\cite{GK07,KG12} on which MV-semirings are based on.
In this framework relations are annotated with elements from an annotation domain $K$. Depending on the domain $K$, the annotations can serve different purposes. For instance, 
 natural number  annotations ($\mathbb{N}$) represent the multiplicity of tuples under bag semantics while using polynomials over a set of variables (e.g., $x_1$, $x_2$, \ldots) representing tuple identifiers the annotations encodes provenance. 
Let $\semK = (K,+_{\semK},\times_{\semK},0_{\semK},1_{\semK})$ be a commutative semiring. A \textbf{$\semK$-relation} $R$ is a (total) function that maps tuples to elements from $\semK$ with the convention that tuples
mapped to $0_{\cal K}$, the $0$ element of the semiring, are not in the
relation. A structure $\semK$ is a commutative semiring if it fulfills the equational laws shown on the top of Figure~\ref{fig:congr-for-mvK}. 
As we will see in the following, the operators of the positive relational algebra ($\RAPlus$) over
$\semK$-relations are defined by combining input annotations using the $+_\semK$ and $\times_\semK$ operations where addition represents alternative use of inputs (e.g., union) and multiplication denotes conjunctive use of inputs (e.g., join).
%
%
 The
semiring $\mathbb{N}$, the set of natural numbers with standard arithmetics
corresponds to bag semantics. 
For example, if a tuple $t$ occurs twice in a relation $R$, then this tuple would be annotated with $2$ in the $\mathbb{N}$-relation corresponding to $R$.


\partitle{Provenance polynomials}
Provenance polynomials (semiring $\ppSR$), polynomials over a set of variables
$X$ which represent tuples in the database, model a very expressive type of provenance by encoding how a query result tuple was derived by combining input tuples.   Using $\ppSR$, every tuple in an instance
is annotated with a unique variable $x \in X$ and the results of queries are annotated with polynomials over these variables.
For example, if a tuple was derived by joining input tuples identified by $x_1$ and $x_2$, then it would be annotated with $x_1 \times x_2$.
%
 Since we are mainly concerned with provenance, we mostly limit the discussion to $\ppSR$ and its MV-semiring extension as explained below.


\begin{figure}[t]
  \centering
\textbf{Laws of commutative semirings}\\[-6mm]
  \begin{align*}
k + 0_\semK &= k
&k \times 1_\semK &= k \tag{neutral elements}
\end{align*}\\[-8mm]
\begin{align*}
k + k' &= k' + k
&k \times k' &= k' \times k \tag{commutativity}
\end{align*}\\[-7mm]
\begin{align*}
  \begin{split}
    k + (k' + k'') &= (k + k') + k''\\[-1mm]
k \times (k' \times k'') &= (k \times k') \times k''
  \end{split}  \tag{associtivity}
\end{align*}\\[-6mm]
\begin{align*}
k \times 0_\semK &= 0_\semK \tag{annihilation through $0$}\\[-1mm]
k \times (k' + k'') &= (k \times k') + (k \times k'') \tag{distributivity}
\end{align*}
\textbf{Evaluation of expressions with operands from $K$}\\[-5mm]
  \begin{align*}
k + k' &= k +_{\semK} k' 
&k \times k' &= k \times_{\semK} k' \tag{if $k \in K \wedge k' \in K$}
  \end{align*}
\textbf{Equivalences involving version annotations}\\[-5mm]
  \begin{align*}
\upMark(0_\semK) &= 0_\semK
&\upMark(k + k') &= \upMark(k) + \upMark(k')   
  \end{align*}\\[-2mm]
  \caption{Equivalence relations for $\mvK$}
  \label{fig:congr-for-mvK}
\end{figure}

\partitle{MV-semirings}
In~\cite{AG16} we have introduced MV-semirings which are a specific class of semirings that encode the derivation of tuples based on a history of transactional updates. For each semiring $\semK$, there exists a corresponding semiring $\mvK$, e.g., $\ppSRV$ is the MV-semiring corresponding to the  provenance polynomials semiring $\ppSR$. Since $\mathbb{N}$ encodes bag semantic relations, $\mvOf{\mathbb{N}}$ represents bag semantics with embedded history.  
 Figures~\ref{fig:running-example-instance} and~\ref{fig:updated-example-instance} show examples of $\ppSRV$ annotations on the left of tuples. 
In these  symbolic $\ppSRV$ expressions variables (e.g., $x_1$, $x_2$, \ldots) represent identifiers of freshly inserted tuples and uninterpreted function symbols called \textit{version annotations} encode which operations  (e.g., a relational update) were applied to the tuple. The nesting of version annotations records the sequence of operations that were applied to create a tuple version. For instance, consider the annotation of tuple $e_1$ in Figure~\ref{fig:running-example-instance}. This tuple was inserted at time $2$ by Transaction $\xid_0$ and was assigned an identifier $1$ ($\upMarker{I}{\xid_0}{2}{1}$). The tuple became visible to other transactions after $\xid_0$'s commit ($\upMarker{C}{\xid_0}{6}{1}$). Observe that these annotations encode what operations have been applied to tuples and from which other tuples they were derived.

\partitle{Version Annotations} 
A version annotation $\upMarker{X}{\xid}{\version}{\tid}(k)$ denotes that an operation of type $X$ (one of  update $U$, insert $I$, delete $D$, or commit $C$) that was executed at time $\version - 1$ by transaction $\xid$ did affected a previous version of a tuple with identifier $\tid$ and previous provenance $k$. 
Assuming domains of tuple identifiers $\tidDomain$, version identifiers $\versionDomain$, and transaction identifiers $\xidDomain$, we use $\upMarkers$ to denote the set of all possible version annotations. This set contains the following version annotations for each $\tid \in \tidDomain$, $\version \in \versionDomain$, and $\xid \in \xidDomain$:
\begin{align}\label{eq:set-of-version-annotations}
  \upMarker{I}{\xid}{\version}{id},
  \upMarker{U}{\xid}{\version}{id},
  \upMarker{D}{\xid}{\version}{id},
  \cMarker{\xid}{\version}{id}
\end{align}\\[-9mm]
%

\partitle{MV-semiring Annotation Domain}
In the running example, the derivation history of each tuple is a linear sequence of operations applied to a single previous tuple version. However, in the general case a tuple can depend on multiple input tuples, e.g., a query that projects an input relation onto a non-unique column ($\projection_{Position}(Employee)$) or an update that modifies two tuples that are distinct in the input to be the same in the output (e.g., \lstinline!UPDATE Employee SET!\\
\lstinline! ID = 101, Name = Peter!). In MV-semiring annotations this is expressed by combining the variables representing input tuples using operations $+$ and $\times$ in the expressions. 
Fixing a semiring $\semK$, the domain of $\mvK$ is the set of finite symbolic expressions $P$ defined by the grammar shown below where $k \in K$ and $\upMark \in \upMarkers$.\\[-4mm]
\begin{align}\label{eq:mv-grammar}
P \defas k \mid P + P \mid P \times P \mid \upMark(P)
\end{align}\\[-5mm]
For example, consider a query $\projection_{Position}(Employee)$ evaluated over the instance from Figure~\ref{fig:updated-example-instance}. The result tuple \texttt{(Software Architect)} is derived from ${e_1}'$ or, alternatively, from $e_2$ (the two tuples with this value in attribute \texttt{position}) and, thus, would be annotated
with\\[-5mm]
\begin{align*}
&Employee({e_1}') + Employee(e_2)\\
= &\cMarker{T_{7}}{26}{1}(\upMarker{U}{T_{7}}{21}{1}(\cMarker{T_0}{6}{1}(\upMarker{I}{T_0}{2}{1}(x_1)))) + \cMarker{T_0}{6}{2}(\upMarker{I}{T_0}{3}{2}(x_2))
\end{align*}
We would expect certain symbolic expressions produced by the grammar above to be equivalent, e.g., expressions in the embedded semiring $\semK$ can be evaluated using the operations of the semiring ($k_1 + k_2 = k_1 +_{\semK} k_2$) and updating a non-existing tuple does not lead to an existing tuple ($\upMark(0_{\semK}) = 0_{\semK}$). This is achieved by using $\mvDom$, the set of congruence classes (denoted by $\congr{}$) for expressions in $P$ based on the equivalence relations 
as shown in Figure~\ref{fig:congr-for-mvK}.

\begin{defi}\label{def:mv-k}
 Let $\semK = (K,+_{\semK},\times_{\semK},0_{\semK},1_{\semK})$ be a commutative semiring. 
The MV-semiring $\mvK$ for $\semK$ is the structure 
$$\mvK = (\mvDom,+_{\mvK},\times_{\mvK},\congr{0_\semK}, \congr{1_\semK})$$
 
where $\times_{\mvK}$ and $+_{\mvK}$ are defined as  
\begin{align*}
\congr{k} \times_{\mvK} \congr{k'} &= \congr{k \times k'}
&\congr{k} +_{\mvK} \congr{k'} &= \congr{k + k'}
\end{align*}
\end{defi}

The definition of addition and multiplication has to be read as: create a symbolic expression by connecting the inputs with $+$ or $\times$ and then output the congruence class for this expression.
For example, $k = \upMarker{U}{\xid}{\version}{1}(10 + 5)$ is a valid element of $\mvOf{\mathbb{N}}$, the bag semantics MV-semiring, which denotes that a tuple with identifier $1$ was produced by an update ($U$) of transaction $\xid$ at version $\version$. 
This element $k$ is in the same equivalence class as $\upMarker{U}{\xid}{\version}{1}(15)$ based on the equivalence that enables evaluation of addition over elements from $\semK$.

\partitle{Normal Form and Admissible Instances}
We have shown in~\cite{AG16} that $\mvK$ expressions admit a (non unique) normal form representing an element $k \in \mvDom$ as a sum $\sum_{i = 0}^{n} k_i$ where none of the $k_i$ contains any addition operations. Intuitively, each summand corresponds to a tuple under bag semantics. Thus, we will sometimes refer to a summand as a tuple version in the following.
Assuming an arbitrary, but fixed, order over such summands we can address elements in such a sum by position. Following~\cite{AG16} we use $\numInSum{k}$ to denote the number of summands in a normalized annotation $k$ and $\nthOfK{k}{i}$ to refer to the i\textsuperscript{th} element in the sum according to the assumed order. In the definition of updates we will make use of this normal form. Note that not all expressions produced by the grammar in Equation~\eqref{eq:mv-grammar} can be produced by transactional histories. For instance, $\upMarker{U}{\xid}{3}{1}(\upMarker{C}{\xid}{2}{1}(\ldots)$ can never be produced by any history, because it would imply that an update of transaction $\xid$ was applied after the transaction committed. An \textbf{admissible} $\mvK$ database instance is defined as an instance that is the result of applying a transactional history (to be defined later in this section) to an empty input database.

\begin{exam}
  Consider the $\NXv$-relation \texttt{Bonus} from the
  example shown in Figure~\ref{fig:updated-example-instance}. 
  The first tuple ${b_1}'$  is annotated with
$\cMarker{T_{7}}{26}{4}(\upMarker{U}{T_{7}}{22}{4}$ $(\cMarker{T_1}{10}{4}(\upMarker{I}{T_1}{8}{4}(x_4))))$,
  i.e., it was created by an update of Transaction $\xid_7$, that updated a
  tuple inserted by $\xid_1$. Based on the outermost commit annotation we know that this tuple version is visible to transactions starting after version $25$.
%
   We use the relational encoding of $\mvK$-relations from~\cite{AG16} restricted to tuples affected by a given transaction to be able to compute provenance using a regular DBMS and to limit provenance to a transaction of interest for a user.
  Figure~\ref{fig:rel-enc-example} shows the relational encoding of $Bonus$
  restricted to the part of the history corresponding to transaction $\xid_{7}$. We abbreviate relation \texttt{Bonus} as $B$. Version annotations 
  are represented as boolean attributes ($\upAttr_i$ for
  update $u_i$) which are \textit{true} if this part of the provenance has this
  \textit{version annotation} and false otherwise. The
  attributes $\upAttr_1$ and $\upAttr_2$ represent the version annotations for
  the first update ($u_1$) and second update ($u_2$) of $T_{7}$.
The only tuple in the instance represents the
  annotation of tuple (1,101, 2000). The annotation contains only a single version
  annotation $\upMarker{U}{T_{7}}{22}{4}$. Thus,
  only the attribute $\upAttr_2$ for update $u_2$ 
  corresponding to this version annotation is true and the other attribute encoding a
  version annotation is set to false. Variables are encoded as the input tuple annotated with the variable ($b_1$ in the example). 
\end{exam}
\partitle{Queries and Update Operations}
\label{sec:update-and-trans}
We use the definition of positive relational algebra ($\RAPlus$) over $\semK$-relations of~\cite{AG16}. 
Let $t.A$ denote the projection of a
tuple $t$ on a list of projection expressions $A$ and $t[R]$ to denote
the projection of a tuple $t$ on the attributes of relation $R$. For
a condition $\theta$ and tuple $t$, $\theta(t)$ denotes a function
that returns $1_{\semK}$ if $t \models \theta$ and $0_{\semK}$ otherwise. 

\begin{defi}
Let 
$R$ and $S$ denote $\semK$-relations, $\schema{R}$ denote the schema of relation $R$, $t$, $u$ denote tuples, and $k \in K$. 
The operators of $\RAPlus$ on $\semK$-relations are defined as:
\begin{align*}
  \projection_A(R)(t) &= \sum_{u: u.A = t} R(u)
  &(R \union S)(t) &= R(t) + S(t)
\end{align*}\\[-8mm]
\begin{align*}
\selection_\theta(R)(t) &= R(t) \times \theta(t)
 & \asingleton{t'}{k}(t) &=
  \begin{cases}
    k & \mathtext{if} t = t'\\
    0_{\semK} & \mathtext{else}
  \end{cases} 
\end{align*}\\[-8mm]
\begin{align*}
(R \join S)(t) &= R(t[R]) \times S(t[S]) \tag{for any $\schema{R} \union \schema{S}$ tuple $t$} 
\end{align*}
\end{defi}


Updates are also defined using the operations of the MV-semiring, but updates add new version annotations to previous annotations.
%
%
The supported updates correspond to SQL constructs \lstinline!INSERT!, \lstinline!UPDATE!, and \lstinline!DELETE!, and \lstinline!COMMIT!.
An operation is executed at a time $\version$ as part
of a transaction $\xid$. 
Update operations take as input a normalized, admissible $\mvK$-relation $R$ and
return the updated version of this $\mvK$-relation. 
An insertion $\relInsert{Q}{\xid}{\version}{R}$ inserts the result of
query $Q$ into relation $R$. 
The annotations of inserted 
tuples are wrapped in version annotations and are assigned fresh tuple identifiers ($\tid_{new}$).
An update operation
$\relUpdate{\theta}{A}{\xid}{\version}{R}$ applies the projection expressions in $A$ to each tuple that fulfills condition $\theta$. 
Both $\relUpdate{\theta}{A}{\xid}{\version}{R}$ and $\relDelete{\theta}{\xid}{\version}{R}$ wrap the annotations of all tuples
fulfilling condition $\theta$ in version annotations.
A commit  $\commitOp{\xid}{\version}{R}$ adds commit version annotations. 
\begin{defi}\label{def:updates}
Let $R$ be an admissible $\mvK$-relation. 
We use $\version(\up)$ to denote the version (time) when an update $\up$ was executed and $\tidOf(k)$ to denote the id of the outermost version annotation of $k \in \mvDom$. Let $A$ be a list of projection expressions with the same arity as $R$, and $\tid_{new}$ to denote a fresh id that is deterministically created as discussed below. Let $Q$ be a query over a database $D$ such that for every $\asingleton{t}{k}$ operation in $Q$ we have $k \in \semK$. The update operations on $\mvK$-relations are defined as:
\begin{align*}
&\relUpdate{\theta}{A}{\xid}{\version}{R}(t) =
    R(t) \times (\neg \theta)(t) \\ &\mathtab\mathtab\mathtab + \sum_{u: u.A = t} \sum_{i=0}^{\numInSum{R(u)}}
    \upMarker{U}{\xid}{\version+1}{\tidOf(\nthOfK{R(u)}{i})}(\nthOfK{R(u)}{i}) \times \theta(u)
\end{align*}\\[-8mm]
\begin{align*}
\relInsert{Q}{\xid}{\version}{R}(t) &= R(t) + \upMarker{I}{\xid}{\version+1}{\tid_{new}}(Q(D)(t))
\end{align*}\\[-8mm]
\begin{align*}
\relDelete{\theta}{\xid}{\version}{R}(t) &= R(t) \times (\neg \theta)(t)\\ &\mathtab +
\sum_{i = 0}^{\numInSum{R(t)}} \upMarker{D}{\xid}{\version+1}{\tidOf(\nthOfK{R(t)}{i})}(\nthOfK{R(t)}{i}) \times \theta(t)
\end{align*}\\[-8mm]
\begin{align*}
\commitOp{\xid}{\version}{R}(t) &= \sum_{i = 0}^{\numInSum{R(t)}} 
\doCommit{\xid}{\version}{\nthOfK{R(t)}{i}}\\
\end{align*}\\[-4mm]
\resizebox{1\linewidth}{!}{
\begin{minipage}{1.0\linewidth}
\begin{align*}
\doCommit{\xid}{\version}{k} &=
                               \begin{cases}
                                 \cMarker{\xid}{\version+1}{\tid}(k) &\mathtext{if} k = \upMarker{I/U/D}{\xid}{\version'}{\tid}(k')\\
                                 k &\mathtext{else}
                               \end{cases}
\end{align*}
\end{minipage}
}
\end{defi}

As a convention, if an attribute $a$ is not listed in the list of expressions $A$ of an update then $a \to a$ is assumed.
For instance, abbreviating \textit{Software Architect} as \textit{SA} the first update of example transaction $\xid_7$ would be written as
$$\relUpdate{ID=101}{\text{\upshape \textquotesingle}SA\text{\upshape \textquotesingle} \to position}{\xid_7}{20}{Employee}$$
What tuple identifiers are assigned by inserts to new tuples is irrelevant as long as identifiers are deterministic and fulfill certain uniqueness requirements. Thus, we ignore identifier assignment here (see~\cite{AG16} for a detailed discussion).

\BG{REMOVE HOMOMORPHISMS completely?}
\section{Challenges and Contributions}
Adapting the MV-semiring model and reenactment approach to RC-SI is challenging, because the visibility rules of RC-SI are more complex than SI, i.e., different statements within a transaction can see different snapshots of the database. Under SI, the first statement of a transaction $\xid$ sees a snapshot as of the time when $\xid$ started and later statements see the same snapshot and the modifications of previous updates from the same transaction. Under RC-SI, each statement $\up$ also sees modifications of earlier updates from the same transaction, but in addition sees updates of concurrent transactions that committed before $\up$ executed. This greatly complicates the definition of transactional semantics in the MV-semiring model. However, as we will demonstrate it is possible to define RC-SI semantics for MV-annotated databases without extending the annotation model
(only the visibility rules have to be adapted).
Under SI reenactment, queries for individual updates can simply be chained together to construct the reenactment query for a transaction. However, a naive extension of this idea to RC-SI would require us to generate the version of the database seen by a certain statement $\up$ by carefully merging the snapshot of the database at the time of $\up$'s execution with the previous changes by $\up$'s transaction. Thus, while the SI reenactment query for a transaction has to read each updated relation only once, a naive approach for RC-SI would have to read each relation $R$ once for each update that affected it. We present a solution that only has to read each relation once in most cases. 
Consequently, it significantly reduces the complexity of provenance computation for RC-SI transactions. 
%
The main contributions of this work are:

\begin{compactitem}
\item We extend the \textbf{multi-version provenance model}, a provenance model for database queries, updates, and transactions to support RC-SI concurrency control protocol (Section~\ref{sec:history}). 

\item We extend our \textbf{reenactment} approach to support computing provenance of RC-SI workloads and present several novel optimizations that are specific to RC-SI including a technique for reducing the number of relation accesses in reenactment (Section~\ref{sec:reenactment-queries}).




\item Our experimental evaluation demonstrates that reenactment for RC-SI is efficient and scales to large da\-ta\-bases and complex workloads (Section~\ref{sec:experiments}).
\end{compactitem}



\section{Related Work}
\label{sec:related-work}

Green et al.~\cite{GK07} have introduced provenance polynomials and the semiring annotation model which generalizes several other provenance models for positive relational algebra including Why-provenance, minimal Why-provenance~\cite{BK01}, and Lineage~\cite{CW00b}. This model has been studied intensively covering diverse topics such as 
relations annotated with annotations from multiple semirings~\cite{KB12}, rewriting queries to minimize provenance~\cite{AD11b}, factorization of provenance polynomials~\cite{OZ11}, extraction of provenance polynomials from the PI-CS~\cite{GM13} and Provenance Games~\cite{KL13} models, and extensions to set difference~\cite{GP10} and aggregation~\cite{AD11d}.
Systems such as DBNotes~\cite{BC05a}, LogicBlox~\cite{GA12}, Perm~\cite{GM13}, Lipstick~\cite{AD11c}, and others encode provenance annotations as standard relations and use query rewrite techniques to propagate these annotations during query processing. 
Many use cases such as auditing and post-mortem transaction debugging require provenance for update operations and particularly transactions. 
In~\cite{AG16,AG14}, we have introduced an extension of the semiring model for SI transactional histories that is the first provenance model supporting concurrent transactions and have pioneered the reenactment approach for computing such provenance over regular relational databases.
Several papers~\cite{BC08a,VC07,AD13a} study provenance for updates, e.g., 
Vansummeren et al.~\cite{VC07} compute provenance for SQL DML statements. This approach alters updates to eagerly compute provenance. 
However, developing a provenance model for transactional updates is more challenging as it requires to consider the complex interdependencies between tuple
versions that are produced by concurrent transactions under different isolation levels. In this work we present the non-trivial extensions of our previous approach~\cite{AG16,AG14} to efficiently support the RC-SI protocol which is widely used in practice. 
\section{Read-Committed SI Histories} 
\label{sec:history}

\begin{figure*}[t]
  \centering

\subfloat[Historic relation $\relV{\rel}{\xid}{\version}$: version of $\rel$ seen by Transaction $\xid$ at Time $\version$]{\label{fig:si-hist-db-definitions-a}
  \begin{minipage}{1.0\linewidth}
\vspace{-3mm}
    \begin{align*}
      \relV{\rel}{\xid}{\version} &=
      \begin{cases}
        \emptyset & \mathtext{if} \version < \start{\xid}\\
        \relCV{\rel}{\version} & \mathtext{if} \start{\xid} = \version\\
        u(\relVE{\rel}{\xid}{\version -1}) & \mathtext{if} \exists u \in \xid: \version(u) = \version - 1 \wedge u \mathtext{updates} R \wedge \finish{\xid} \neq \version - 1\\
        \commitOp{\xid}{\version - 1}{\relV{\rel}{\xid}{\version - 1}} & \mathtext{if} \finish{\xid} = \version - 1\\
        \relV{\rel}{\xid}{\version - 1} & \mathtext{otherwise} 
      \end{cases}
    \end{align*}
  \end{minipage}
}\\[1mm]
\begin{minipage}{1.0\linewidth}
  \subfloat[$\relVE{\rel}{\xid}{\version}$: Tuple versions visible within Transaction $\xid$ at Time
  $\version$]{\label{fig:si-hist-db-definitions-b}
    \begin{minipage}{1.0\linewidth}
      \vspace{-2mm}
      \begin{align*}
       \relVE{\rel}{\xid}{\version}(t) &= \hspace{-2mm} \sum_{i=0}^{\numInSum{\relCV{\rel}{\version}(t)}} \hspace{-2mm} \nthOfK{\relCV{\rel}{\version}(t)}{i} \times \validEx (\xid,t, \nthOfK{\relCV{\rel}{\version}(t)}{i}, \version)
+ \hspace{-4mm}
\sum_{i=0}^{\numInSum{\relV{\rel}{\xid}{\version}(t)}} \hspace{-4mm} \nthOfK{\relV{\rel}{\xid}{\version}(t)}{i} \times \validIn (\xid,t, \nthOfK{\relV{\rel}{\xid}{\version}(t)}{i}, \version)
      \end{align*}\\
    \end{minipage}
  }
\end{minipage}\\[-4mm]
\begin{minipage}{1.0\linewidth}
  \subfloat[$\relCV{\rel}{\version}$: Committed tuple versions at Time
  $\version$]{\label{fig:si-hist-db-definitions-c}
    \begin{minipage}{1.0\linewidth}
      \vspace{-2mm}
      \begin{align*}
        \relCV{\rel}{\version}(t) &= \sum_{\xid \in \history \wedge \finish{\xid} < \version}
                                    \sum_{i=0}^{\numInSum{\relV{\rel}{\xid}{\version}(t)}} \nthOfK{\relV{\rel}{\xid}{\version}(t)}{i} \times \validAt(\xid, t, \nthOfK{\relV{\rel}{\xid}{\version}(t)}{i}, \version)
      \end{align*}\\
    \end{minipage}
  }
\end{minipage}\\[-4mm]
\begin{minipage}{1.0\linewidth}
  \subfloat[Validity of summands (tuple versions) within annotations]{\label{fig:si-hist-db-definitions-d}
    \begin{minipage}{1.0\linewidth}
      \vspace{-4mm}
      \begin{align*}
        \validIn(\xid,t,k, \version) &= 1 \mathtext{if} \exists \version', k',\tid: k = \upMarker{X}{\xid}{\version'}{\tid}(k') \wedge X \in \{U,D,I\}, 0
                                       \mathtext{otherwise} \\
\validEx(\xid,t,k, \version) &= 0 \mathtext{if} \hasUp(T,t,k, \version), 1 \mathtext{otherwise}
       \end{align*}\\[-10mm]
      \begin{align*}
        \validAt(T,t,k,\version) &= 1 \mathtext{if} k = \cMarker{\xid}{\version'}{\tid}(k') \wedge
          (\neg\exists \xid' \neq \xid: \finish{{\xid'}} \leq \version
          \wedge \hasUp(\xid',t,k, \version)), 0 \mathtext{otherwise}
      \end{align*}\\[-10mm]
      \begin{align*}
        \hasUp(T,t,k,\version) &\Leftrightarrow \exists u \in T, t', i, j :
                        \version(\up) < \version \wedge
                        \nthOfK{\relV{\rel}{\xid}{\version(u)}(t)}{i} = k \wedge
                        \nthOfK{\relV{\rel}{\xid}{\version(u) + 1}(t')}{j} = \upMarker{X}{\xid}{\version(u)+1}{id}(k) \wedge X \in \{U,D\}
      \end{align*}\\[-7mm]
    \end{minipage}
  }
\end{minipage}

  \caption{Historic relational instances induced by History $\history$. $\relV{\rel}{\xid}{\version}$ is the annotated instance visible by Transaction $\xid$ at version $\version$. $\relCV{\rel}{\version}$ is the instance containing all changes of transactions committed before version $\version$. Each update of a transaction sees all modifications of previous updates from the same transaction as well as modifications of transactions committed before the update was run ($\relVE{\rel}{\xid}{\version}$).}
  \label{fig:si-hist-db-definitions}
\end{figure*}

We now define the semantics of RC-SI histories over $\mvK$-relations. Importantly,
our extension uses standard MV-semirings and update operations. 
%
A \textbf{transaction} $\xid = \{\up_1,$ $\ldots,\up_n, c\}$ is a sequence of update
operations followed by a commit operation ($c$)
 with $\version(\up_i) < \version(\up_j)$ for $i < j$.  
A \textbf{history} $\history = \{\xid_1, \ldots, \xid_n\}$ over a database $\db$ is a set of transactions
 over $\db$ with at most one operation at each version $\version$.
We use $\start{\xid} = \version(\up_1)$ and $\finish{\xid} = \version(c)$ to denote the time when transaction $\xid$ did start (respective did commit). Note that the execution order of operations is encoded in the updates itself, because each update $\up$ in the MV-semiring model is associated with a version identifier $\version(\up)$ determining the order of operations. 


Given a RC-SI history $\history$ we define $\relCV{\rel}{\version}$,  the annotated state of relation $\rel$ at a time $\version$ and  $\relV{\rel}{\xid}{\version}$, the annotated state of relation $\rel$ visible to transaction $\xid$ at time $\version$. Note that these two states may differ, because transaction $\xid$'s updates only become visible to other transactions after $\xid$ has committed. 
As in~\cite{AG16} we assume that histories are applied to an empty initial database.  
For instance, Figure~\ref{fig:updated-example-instance} shows a subset of $\dbCV{26}$, the version of the example DB after execution of the history (Figure~\ref{fig:Transactions Example}) over $\dbCV{18}$ (shown in Figure~\ref{fig:running-example-instance}). The database state $\dbCV{18}$ is the result of running Transaction $\xid_0$ that inserted the content of the \texttt{Employee} relation and $\xid_1$, $\xid_2$, $\xid_4$ which created the tuples in relation \texttt{Bonus}. 

\begin{defi}
  Let $\history$ be a history over a database $\db$.  
  The version $\relCV{\rel}{\version}$ of relation $\rel \in \db$  at time $\version$ and the version $\relV{\rel}{\xid}{\version}$ of relation $\rel$ visible within transaction $\xid \in \history$ at time $\version$ are defined in Figure~\ref{fig:si-hist-db-definitions}.

%
\end{defi}


\partitle{Figure~\ref{fig:si-hist-db-definitions-a}: Relation Version in Transaction $\xid$ at Time $\version$}
To define the content of relation $\rel$ at time $\version$ within transaction $\xid$ we have to distinguish between several cases: 
1) per convention $\relV{\rel}{\xid}{\version}$ is empty for any $\version < \start{\xid}$; 
2) at the start of transaction $\xid$, $\relV{\rel}{\xid}{\version}$ is same as $\relCV{\rel}{\version}$, the version of the relation containing changes of transactions committed before $\version$; 
3) if an update was executed by transaction $\xid$ at time $\version-1$ then its effect is reflected in $\relV{R}{\xid}{\version}$. The update will see tuple versions created by transactions that committed before $\version -1$ and tuple versions created by the transaction's own updates. We use $\relVE{\rel}{\xid}{\version-1}$ to denote this version of $\rel$ and  explain its construction below; 
4) right after transaction commit, the current version of the relation visible within $\xid$ is the result of applying the commit operator to the previous version;
and 5) as long as there is no commit or update on $R$ at $\version - 1$ then the current version of relation $\rel$ is the same as the previous one. 

\partitle{Figure~\ref{fig:si-hist-db-definitions-b}: Relation Version Visible to Updates}
As mentioned above we use $\relVE{\rel}{\xid}{\version}$ to denote the version of relation $\rel$ that is visible to an update of transaction $\xid$ executed at time $\version$. This state of relation $\rel$ contains all tuple versions created by committed transactions as long as they have not been overwritten by a previous update of transaction $\xid$ (the first sum) and tuple versions created by previous updates of transaction $\xid$ (the second sum). Here by overwritten we mean that a tuple version is no longer valid, because either it has been deleted or because it was updated and, thus, it has been replaced with a new updated version. Function $\validEx$ implements this check. It returns $1$ if the tuple version has not been overwritten and $0$ otherwise.
This function uses a predicate $\hasUp(T,t,k,\version)$ which is true if transaction $\xid$ has invalidated summand $k$ in the annotation of tuple $t$ before $\version$ by either deleting or updating the corresponding tuple version.
The second sum ranges over tuple versions $\relV{\rel}{\xid}{\version}$ excluding tuple versions not created by transaction $\xid$ (function $\validIn$). 

\partitle{Figure~\ref{fig:si-hist-db-definitions-c}: Committed Relation Version}
The committed version $\relCV{\rel}{\version}$ of a relation $\rel$ at time $\version$ contains all changes of transactions that committed before $\version$. That is, all tuple versions created by any such transaction unless the tuple version is no longer valid at $\version$, e.g., it got deleted by another transaction. Thus, this version of relation $\rel$ can be computed as the sum over all annotations on tuple $t$ in the versions of relation $\rel$ created by past transactions. However, in addition to ensuring that outdated tuple versions are not considered we also need to ensure that every tuple version is only included once. Both conditions are modelled by function $\validAt(\xid,t,k,\version)$ that return $1$ if $k$ is a summand (tuple version) in the annotation of tuple $t$ at time $\version$ and was created by $\xid$ (this ensures that each tuple version is only added once).

\begin{exam}
  Consider the example transactional history from Figure~\ref{fig:Transactions Example}. For instance, $\relV{Bonus}{\xid_8}{22}$ is the version of the \texttt{Bonus} relation seen by the insert operation of Transaction $\xid_8$ and is equal to $\relCV{Bonus}{22}$ (case 2, Figure~\ref{fig:si-hist-db-definitions-a}). It contains the tuples from the \texttt{Bonus} relation as shown in Figure~\ref{fig:running-example-instance}, because these tuples were created by transactions that committed before time $22$ (they are in $\relV{Bonus}{\xid_8}{22}$). Thus, $\validAt$ returns $1$ for these tuples. 
For instance, tuple $b_1$ has been updated by Transaction $\xid_7$ (the new version is denoted as ${b_1}'$) before version $22$, but this transaction has not committed yet. Since $\xid_8$ has not updated $b_1$, $\validEx$ returns $1$ and the full annotation of $b_1$ in $\relCV{Bonus}{22}$ is as shown in Figure~\ref{fig:running-example-instance}.
\end{exam}
\section{Reenactment}
\label{sec:reenactment-queries}

We have introduced reenactment~\cite{AG16} as a mechanism to
construct a $\mvK$-annotated relation $R$ produced by a transaction $\xid$ that is part of a history $\history$ by running a so-called reenactment query
$\ract(\xid)$.  We have proven~\cite{AG16} that $\ract(\xid) \equiv_{\ppSRV} \xid$, i.e., the
reenactment query returns the same annotated relation as the original
transaction ran in the context of history $\history$ (has the same result and
provenance). 
%
In this work, we present reenactment for single RC-SI transactions as well as extensions necessary to reenact a whole history. The latter requires the introduction of a operator which merges the relations produced by the reenactment queries of several transactions. This operator is also needed to compute $\relVE{\rel}{\xid}{\version}$ as introduced in the previous section. After introducing this operator, we first present a method to reenact RC-SI transactions that requires merging newly committed tuples into the version of a relation visible within the reenacted transaction after every update. We then present an optimization that requires no merging in most cases and uses another new operator - version filtering. 



\partitle{Version Annotation Operator}
For reenactment of updates and transactions we need to be able to introduce new version annotations in queries. 
However, the operators of ${\cal RA}^+$ do not support that. To address this problem, we have defined the version annotation operator in~\cite{AG16}.
For $X \in \{I,U,D\}$ the version annotation operator $\annotOp{X}{\xid}{\version}{}(R)$
takes as input a $\mvK$-relation $R$ and wraps every summand in a tuple's annotation in $\upMarker{X}{\xid}{\version}{}$. The commit annotation operator $\annotOp{C}{\xid}{\version}{}(R)$ only wraps summands produced by Transaction $\xid$ using operator $\doCommit{\xid}{\version}{}$ from Definition~\ref{def:updates}.
\begin{align*}
  \annotOp{X}{\xid}{\version}{}(R)(t) =
 \begin{cases}
\doCommit{\xid}{\version}{\nthOfK{R(t)}{i}} &\mathtext{if} X  = C\\
\sum_{i = 0}^{\numInSum{R(t)}} \upMarker{X}{\xid}{\version}{}(\nthOfK{R(t)}{i})   &\mathtext{otherwise}  \\
 \end{cases}
\end{align*}


\partitle{Reenacting Updates}
%
Reenactment queries for transactions are constructed from reenactment queries for single update statements.
The reenactment query $\ract(u)$ for an update $u$ returns the modified version of the relation targeted by the update if it is evaluated over the database state seen by $u$'s transaction at the time of the update $u$  ($\relVE{\rel}{\xid}{\version(u)}$). The semantics of update operations is the same no matter whether SI or RC-SI is applied. Thus, we can use the technique we have introduced for SI in~\cite{AG16} to also reenact RC-SI updates.
As we will see later, it will be beneficial to let update reenactment queries operate over a different input for RC-SI than for SI which requires modifications to the update reenactment queries. 
Let $\history$ is a history over database $D$.
Below we show the definitions of update reenactment  queries  from~\cite{AG16}. The reenactment
query $\ract(\up)$ for operation  $\up$  in $\history$ is:\\[-11mm]
\begin{center}
\resizebox{1\linewidth}{!}{
  \begin{minipage}{1.0\linewidth}
\begin{align*}
\ract(\relUpdate{\theta}{A}{\xid}{\version}{R}) &= \annotOp{U}{\xid}{\version+1}{}(\projection_{A}(\selection_\theta(\relV{R}{\xid}{\version}))) \union \selection_{\neg \theta}(\relV{R}{\xid}{\version})\\ 
\ract(\relInsert{Q}{\xid}{\version}{R}) &= \relV{R}{\xid}{\version} \union \annotOp{I}{\xid}{\version+1}{}(Q(\dbV{\xid}{\version})) \\
\ract(\relDelete{\theta}{\xid}{\version}{R}) &= \annotOp{D}{\xid}{\version+1}{}(\selection_\theta(\relV{R}{\xid}{\version})) \union \selection_{\neg \theta}(\relV{R}{\xid}{\version})
\end{align*}
\end{minipage}
}
\end{center}

For example, an update modifies a relation by applying the expressions from $A$ to
 tuples that match the update condition $\theta$. All other tuples are not affected. Thus, the result of an update can be computed as the
union between these two sets. 
%
%
For instance, the 
  reenactment query $\ract(u_2)$ for the update
  $u_2$ of running example transaction $\xid_7$ is:
  \begin{align*}
    &\annotOp{U}{\xid_7}{21}{}(\projection_{ID,EmpID,Amount+1000 \to Amount}(\\&\mathtab\mathtab\mathtab\selection_{ID=101}(\relV{Bonus}{\xid_7}{21})))
    \union \selection_{ID \neq 101}(\relV{Bonus}{\xid_7}{21})
  \end{align*}
\partitle{Transaction and History Reenactment}
\label{sec:trans-reen}
To reenact a transaction $\xid$, we have to connect reenactment queries for the updates of $\xid$ such that the input of every update $u$ over relation $R$ is $\relVE{R}{\xid}{\version(\up)}$. As discussed in Section~\ref{sec:history}, this instance of relation $R$ contains tuple versions updated by previous updates of $\xid$ which targeted $R$ as well as tuple versions from $\relCV{R}{\version}$. Hence, $\relVE{R}{\xid}{\version(\up)}$ can be computed as a union between these two sets of tuples as long as we 
can filter out tuple versions (summands in annotations) that are no longer valid. We now introduce a new query operator that implements this filtering and then define  reenactment for RC-SI transactions using this operator.

\partitle{Version Merge Operator}
The version merge operator $\vMerge(R_1,R_2)$ is used to merge two version $R_1$ and $R_2$ of a relation $R$ such that 1) tuple versions (summands in annotations) present in both inputs are only included once in the output and 2) if both inputs include different versions of a tuple, then only the newer version is returned. 
This operator is used to construct $\relVE{R}{\xid}{\version}$ from a union of $\relCV{R}{\version}$ and $\relV{R}{\xid}{\version}$.
The definition of $\vMerge(R_1,R_2)$ is shown below. 
\begin{align*}
      \vMerge(R_1,R_2)(t) &= \sum_{i=0}^{\numInSum{R_1(t)}} R_1(t)[i] \times \isMax(R_2,R_1(t)[i]) \\ &+\sum_{i=0}^{\numInSum{S(t)}} R_2(t)[i] \times \isStrictMax(R_1,R_2(t)[i])\\
 \end{align*}\\[-8mm]
The operator uses two functions $\isMax$ and $\isStrictMax$. $\isMax(R,k)$ returns $0$ if relation $R$ contains a newer version of the tuple version encoded as annotation $k$, i.e., if $\exists t',k',j : \idOf(R(t')[j]) = \idOf(k) \wedge \versionOf(R(t')[j]) > \versionOf(k)$. Function $\isStrictMax$ is the strict version of $\isMax$ which also returns $0$ if the tuple version $k$ is present in $R$, i.e., $\versionOf(R(t')[j]) > \versionOf(k)$ is replaced with $\versionOf(R(t')[j]) \geq \versionOf(k)$ in the condition. Here function $\idOf(k)$ returns the tuple identifier in the annotation $k$ and  $\versionOf$  returns the version encoded in the given annotation $k$. These functions are well defined if $k$ is a summand in a normalized admissible $\mvK$-relation (see Section~\ref{sec:prov-model-vers}):
 \begin{align*}
         \idOf(\upMarker{X}{\xid}{\version}{\tid}(k')) &=  \tid 
      &\versionOf(\upMarker{X}{\xid}{\version}{\tid}(k')) &= \version 
 \end{align*}
As an example consider computing $\vMerge(\relCV{Bonus}{26}, \relCV{Bonus}{19})$. These relation versions are shown in Figure~\ref{fig:running-example-instance} and~\ref{fig:updated-example-instance}. The later only shows new or updated tuples. For instance, $b_2$ is present in both relations with the same annotation, a single summand. Thus, the first sum in $\vMerge(\relCV{Bonus}{26}, \relCV{Bonus}{19})(b_2)$ will include this annotation (there is no newer version of this tuple in $\relCV{Bonus}{19}$) while it will be excluded from the second sum (the same annotation is found in $\relCV{Bonus}{26}$). As another example consider tuple $b_1$ which was updated to ${b_1}'$ by Transaction $\xid_7$. Thus, $\vMerge(\relCV{Bonus}{26}, \relCV{Bonus}{19})(b_1) = 0$, because a newer version of this tuple exists in $\relCV{Bonus}{26}$ and $\vMerge(\relCV{Bonus}{26}, \relCV{Bonus}{19})({b_1}')  = \relCV{Bonus}{26})({b_1}')$ (this is the newest version of this tuple found in $\relCV{Bonus}{19}$ and $\relCV{Bonus}{26}$).

\partitle{Reenacting Transactions}
For simplicity of exposition we present the construction of reenactment queries for transactions updating a single relation $\rel$. The construction for transactions updating multiple relations is achieved analog to~\cite{AG16}. The reenactment query for Transaction $\xid = (u_1, \ldots, u_n, c)$ executed as part of an RC-SI history $\history$ is recursively constructed starting with a commit annotation operator applied to the reenactment query $\ract(u_n)$ for the  last update of $\xid$. Then we replace $\relV{R}{\xid}{\version(u_n)}$ in the query constructed so far with $\vMerge(\ract(u_{n-1}),\relCV{R}{\version(u_n)})$. The result of this version merge operator is $\relVE{R}{\xid}{\version(u_n)}$, the input seen by $u_n$ in the history $\history$.  This replacement process is repeated for $i \in {n-1, \ldots, 1}$  until every reference to a version of relation $R$ visible within the transaction has been replaced with references to committed relation versions ($\relCV{R}{\version}$ for some $\version$). The  structure of the reenactment query is outlined below.\\[-7mm]
\begin{center}
  \resizebox{1\linewidth}{!}{
    \begin{tikzpicture}[
nonmat/.style={rectangle,draw=black!100,fill=black!10,thick},
mat/.style={rectangle,draw=red!70,fill=red!10,thick},
qop/.style={circle,draw=red!70,fill=red!10,thick},
trans/.style={->,shorten <=1pt,,semithick},
dep/.style={->,dotted,shorten <=1pt,semithick},
labl/.style={blue}
]
\node (R1) at (-3,0) [nonmat] {$\relCV{R}{\version(u_1)}$};
\node (RU1) at (-1,0) [mat] {$\ract(u_1)$};

\node (R2) at (-1,-1) [nonmat] {$\relCV{R}{\version(u_2)}$};
\node (VM2)at (0.25,0) [qop] {$\vMerge$};
\node (RU2) at (1.5,0) [mat] {$\ract(u_2)$};

\node (R3) at (1.5,-1) [nonmat] {$\relCV{R}{\version(u_3)}$};
\node (VM3)at (2.75,0) [qop] {$\vMerge$};
\node (RU3) at (4,0) [mat] {$\ract(u_3)$};

\node (RUn1) at (7,0) [mat] {$\ract(u_{n-1})$};
\node (Rn) at (7,-1) [nonmat] {$\relCV{R}{\version(u_n)}$};
\node (VMn)at (8.4,0) [qop] {$\vMerge$};
\node (RUn) at (9.5,0) [mat] {$\ract(u_n)$};

\draw (R1) edge [trans] (RU1);

\draw (RU1) edge [trans] (VM2);
\draw (R2) edge [trans] (VM2);
\draw (VM2) edge [trans] (RU2);

\draw (RU2) edge [trans] (VM3);
\draw (R3) edge [trans] (VM3);
\draw (VM3) edge [trans] (RU3);

\draw (RU3) edge [dep] (RUn1);

\draw (RUn1) edge [trans] (VMn);
\draw (Rn) edge [trans] (VMn);
\draw (VMn) edge [trans] (RUn);

    \end{tikzpicture}
  }
\end{center}





\partitle{Reducing Relation Accesses}
We would like reenactment queries for RC-SI to be defined recursively without
requiring to recalculate the right mix of tuple versions from transaction $\xid$
and from concurrent transactions after each update. 
To this end we introduce the version filter operator, that filters out summands $k$ from an annotation based on the version encoded in the outermost version annotation of $k$.
The filter condition $\theta$ of a version filter operator is expressed using a pseudo attribute $V$ representing the $\version$ encoded in version annotations.  We use this operator to filter summands from annotations based on
the version annotations they are wrapped in.


\BG{tuple-ids are enough, We can get rid of all model extensions.}



\partitle{Version Filter Operator}
The version filter operator removes summands from an annotation based on the time $\version$ in their outermost version annotation. 
  Let $\theta$ be a condition over 
  pseudo
  attribute $V$
  . Given a summand $k =
  \upMarker{X}{\xid}{\version}{i}(k')$ such a condition is
  evaluated by replacing $V$ with $\version$ in $\theta$. The version filter operator using such a condition $\theta$
  is defined as:
  \begin{align*}
    \vFilt{\theta}(R)(t) &= \sum_{i=0}^{\numInSum{R(t)}} R(t)[i] \times \theta(R(t)[i])
  \end{align*}

For example, we could use $\vFilt{V < 11}(R)$ to filter out summands from annotations of tuples from a relation $R$ that were added after time $10$. In contrast to regular selection, a version filter's
condition is evaluated over the individual summands in an annotation. 

Our optimized reenactment approach for RC-SI is based on the following observation. Consider a tuple $t$ updated by Transaction $\xid$ and let $u \in \xid$ be the first update of Transaction $\xid$ that modified this tuple. 
Let $t'$ denote the version of tuple $t$ valid before $u$. Given the RC-SI semantics, $t'$ is obviously present in $\relCV{\rel}{\version(u)}$ and was produced by a transaction that committed before $\version(u)$. Importantly, $t'$ is guaranteed to be in $\relCV{\rel}{\finish{\xid}}$, i.e, the version of $\rel$ immediately before the commit of Transaction $\xid$. To see why this is the case recall that $\xid$ would have obtained a write-lock on this tuple to be able to update $t'$ to $t$ and this write-lock is held until transaction commit. Thus, it is guaranteed that no other transaction would have been able to update $t'$ before the commit of $\xid$. Based on this observation, we can use $\relCV{\rel}{\finish{\xid}}$ as an input to the reenactment query as long as we ensure that the reenactment queries for other updates of $\xid$ executed before $u$ ignore $t'$. We achieve this using the version filter operator to filter out tuple versions that were not visible to an update $u'$. It is applied in the input of the part of the transaction reenactment query corresponding to the update $u'$. In the optimized reenactment query, the initial input of reenactment is $\relCV{\rel}{\finish{\xid}}$ instead of $\relCV{\rel}{\start{\xid}}$. 
Furthermore, the update reenactment queries are modified as shown below. An optimized reenactment query $\ract_{opt}(u)$ for update $u$ passes on unmodified versions  of tuples that are not visible to update $u$. We use $\ract_{opt}(\xid)$ to denote the optimized transaction reenactment query. In the formulas shown below, $R$ denotes the result of the reenactment query for the previous update or $\relCV{\rel}{\finish{\xid}-1}$ (in case the update is the first update of the transaction). Note that this optimization is only applicable if the inserts in the transaction do not access the relation that is modified by the updates and deletes of the transaction. That is because the query of an insert may read tuple version that are not in $\relCV{D}{\finish{\xid}}$. Hence, we only apply this optimization if the inserts of Transaction $\xid$ use the \lstinline!VALUES! clause (the singleton operator $\asingleton{t}{k}$ as defined in Section~\ref{sec:prov-model-vers}).
%
\begin{align*}
\ract_{opt}(\relUpdate{\theta}{A}{\xid}{\version}{R}) &= \annotOp{U}{\xid}{\version+1}{}(\projection_{A}(\selection_\theta(\vFilt{V \leq \version(u)}(R))))\\ 
                                                  &\mathtab\union \selection_{\neg \theta}(\vFilt{V \leq \version(u)}(R)) \\
                                                  &\mathtab\union \vFilt{V > \version(u)}(R))\\ 
\ract_{opt}(\relDelete{\theta}{\xid}{\version}{R}) &= \annotOp{D}{\xid}{\version+1}{}(\selection_\theta(\vFilt{V \leq \version(u)}(R)))\\ 
                                               &\mathtab\union \selection_{\neg \theta}(\vFilt{V \leq \version(u)}(R))\\
                                               &\mathtab\union \vFilt{V > \version(u)}(R)
\end{align*}

For example, the reenactment query for an update $u$ distinguishes between three disjoint cases: 1) a tuple that is visible to the update ($V \leq \version(u)$) and fulfills the update's condition, i.e., the tuple is updated by $u$; 2) a tuple that is visible to the update, but does not fulfill the condition $\theta$; and 3) a tuple version that is not visible to $u$, because it was created by a transaction that committed after $\version(u)$. The structure of the resulting reenactment query 
for transactions without inserts
is shown below. Note that relation $\rel$ is only accessed once by the reenactment query. 

\begin{center}
  \resizebox{1\linewidth}{!}{
    \begin{tikzpicture}[
nonmat/.style={rectangle,draw=black!100,fill=black!10,thick},
mat/.style={rectangle,draw=red!70,fill=red!10,thick},
qop/.style={circle,draw=red!70,fill=red!10,thick},
trans/.style={->,shorten <=1pt,,semithick},
dep/.style={->,dotted,shorten <=1pt,semithick},
labl/.style={blue}
]
\node (R1) at (-3,0) [nonmat] {$\relCV{R}{\finish{\xid} - 1}$};
\node (RU1) at (-0.5,0) [mat] {$\ract_{opt}(u_1)$};

\node (RU2) at (1.5,0) [mat] {$\ract_{opt}(u_2)$};

\node (RU3) at (4,0) [mat] {$\ract_{opt}(u_3)$};

\node (RUn1) at (7,0) [mat] {$\ract_{opt}(u_{n-1})$};
\node (RUn) at (9.5,0) [mat] {$\ract_{opt}(u_n)$};

\draw (R1) edge [trans] (RU1);

\draw (RU1) edge [trans] (RU2);

\draw (RU2) edge [trans] (RU3);

\draw (RU3) edge [dep] (RUn1);

\draw (RUn1) edge [trans] (RUn);

    \end{tikzpicture}
  }
\end{center}

For each insert using the \lstinline!VALUES! clause a new tuple will be added to the relation $\rel$ using \lstinline!UNION!.

Reenactment queries for RC-SI transactions are  equivalent to the transaction
they are reenacting.

\begin{theo}\label{theo:update-reenactment-equivalence-RCSI}
  Let $\xid$ be a RC-SI transaction. 
  Then, 
$\xid \equiv_{\ppSRV} \ract(\xid) \equiv_{\ppSRV} \ract_{opt}(\xid)$.
\end{theo}
\begin{proof}
The proof is shown in Appendix~\ref{sec:proofs}.
\end{proof}
To create a reenactment query for a (partial) history, we combine the results of reenactment queries for all transactions in  the history using the version merge operator. Each reference to a committed version of a relation $\relCV{\xid}{\version}$ is replaced with a multiway merge of the results of reenactment queries for transactions $\xid \in \history$ that committed before $\version$ in the order of commit. For example, if two transactions $\xid_1$ and $\xid_2$ have committed before $\version$ then $\relCV{\rel}{\version}$ is computed as
$$q = \vMerge(\ract(\xid_1),\ract(\xid_2))$$
Later versions can then be computed by reusing this query result, e.g., if the next transaction to commit in the history was $\xid_3$, then the version of $R$ at $\finish{\xid_3} + 1$ is computed as $\vMerge(q, \ract(\xid_3))$.

\section{Implementation}
\label{sec:implementation}

GProM is a middleware that implements reenactment for SI over standard DBMS using a relational encoding of MV-relations~\cite{AG16,AG14}. Reenactment is implemented as SQL queries over this encoding. We have extended the system to implement RC-SI reenactment using the same relational encoding. One advantage of this system is that provenance requests are
considered as queries and can be used as subqueries in an SQL statement, e.g., to  query or store provenance. In Section~\ref{sec:experiments} we study the performance of queries over provenance.
GProM assumes that the underlying database system on which we want to execute
provenance computations keeps an audit log that can be queried and provides at least the information as shown in Figure~\ref{fig:Transactions Example}. Furthermore, the DBMS has to support time travel for the system to query past states of relations (this is used to reenact single transactions and partial histories). For instance, Oracle, DB2, and MSSQL support both features.
While a full description of the implementation and additional optimizations is beyond the scope of this paper, we give a brief overview of the additional optimizations that we have implemented: 1) as we observed in~\cite{AG16}, reenactment queries can contain a large number of union operations that may lead to bad performance if they are unfolded by the DBMS. We extend our approach for using \lstinline!CASE! to avoid union operations~\cite{AG16} to RC-SI; 2) if the user is only interested in the provenance of tuples modified by a particular transaction, then this can be supported by filtering tuples from the output of the transaction's reenactment query that were not affected by the transaction. We did present two methods for improving the efficiency of this filter step by either removing tuples from the input of the reenactment query which do not fulfill the condition of any update of the transaction or by retrieving updated tuple versions from  the database version  after transaction commit and using this set to filter the input using a join. We have adapted both methods for RC-SI; 3) the version merge operator is implemented using aggregation to determine the latest version of each tuple. 

\begin{figure*}[t]
  \begin{minipage}[b]{0.245\linewidth}
    \includegraphics[width=1\linewidth,trim=0 50pt 0 80pt, clip]{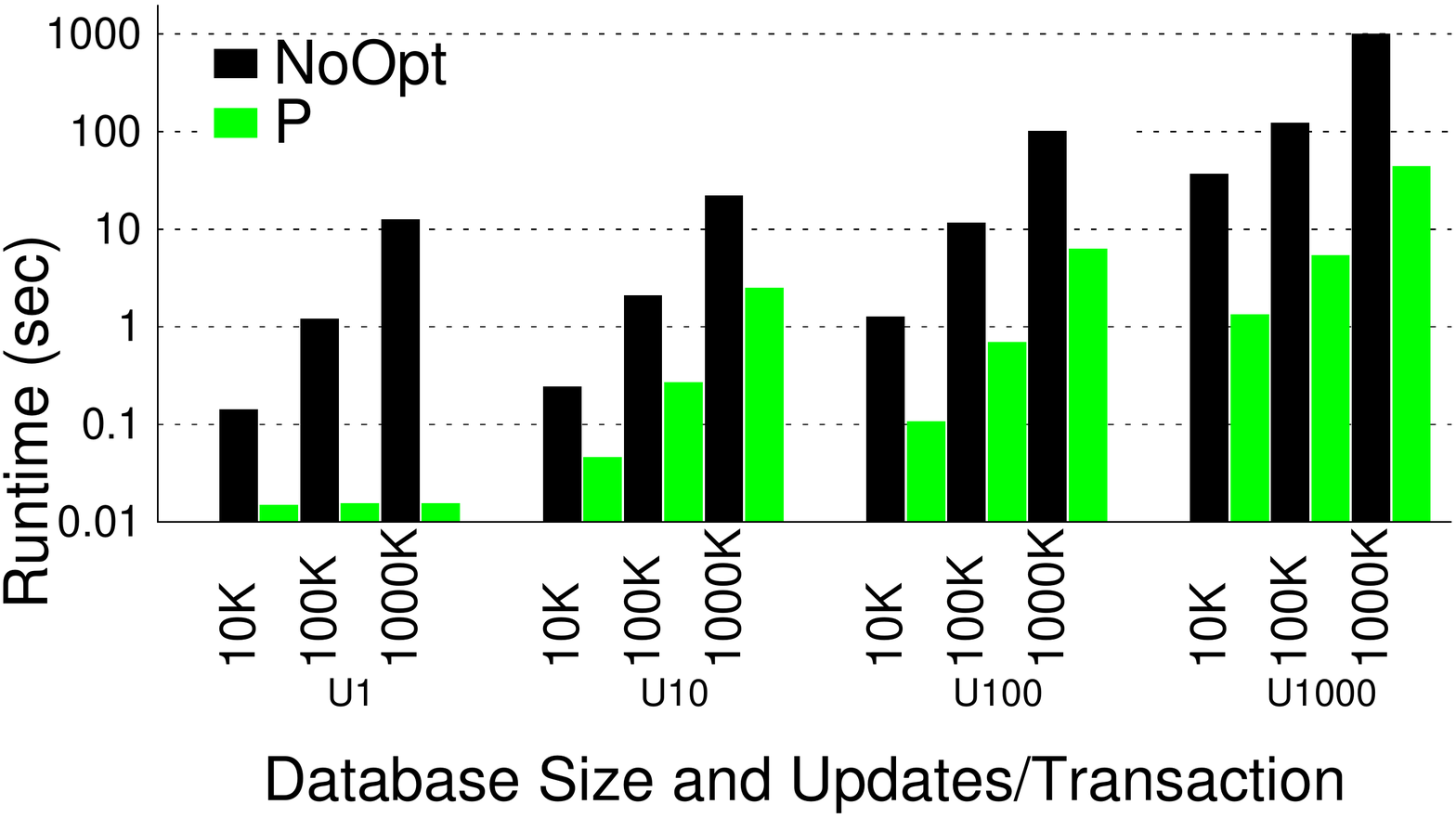}\\
    \vspace{-6.5mm} 
  \caption{Relation Size}
  \label{fig:No-Significant-History}  
  \end{minipage}
  \begin{minipage}[b]{0.245\linewidth}
  \includegraphics[width=1\linewidth,trim=0 70pt 0 80pt, clip]{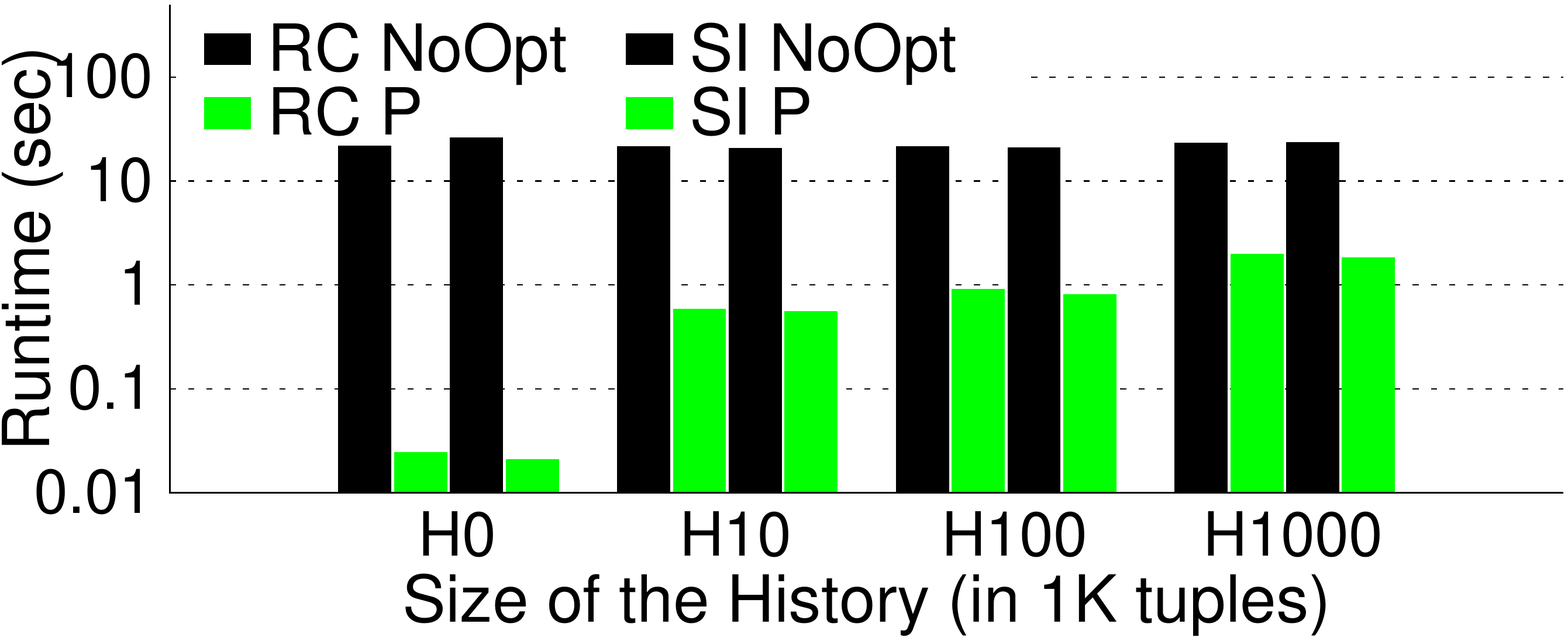}\\
  \vspace{-4mm}
  \caption{History Size}
  \label{fig:Compare-History-Size}  
  \end{minipage}
  \begin{minipage}[b]{0.245\linewidth}
  \includegraphics[width=1\linewidth,trim=0 70pt 0 80pt, clip]{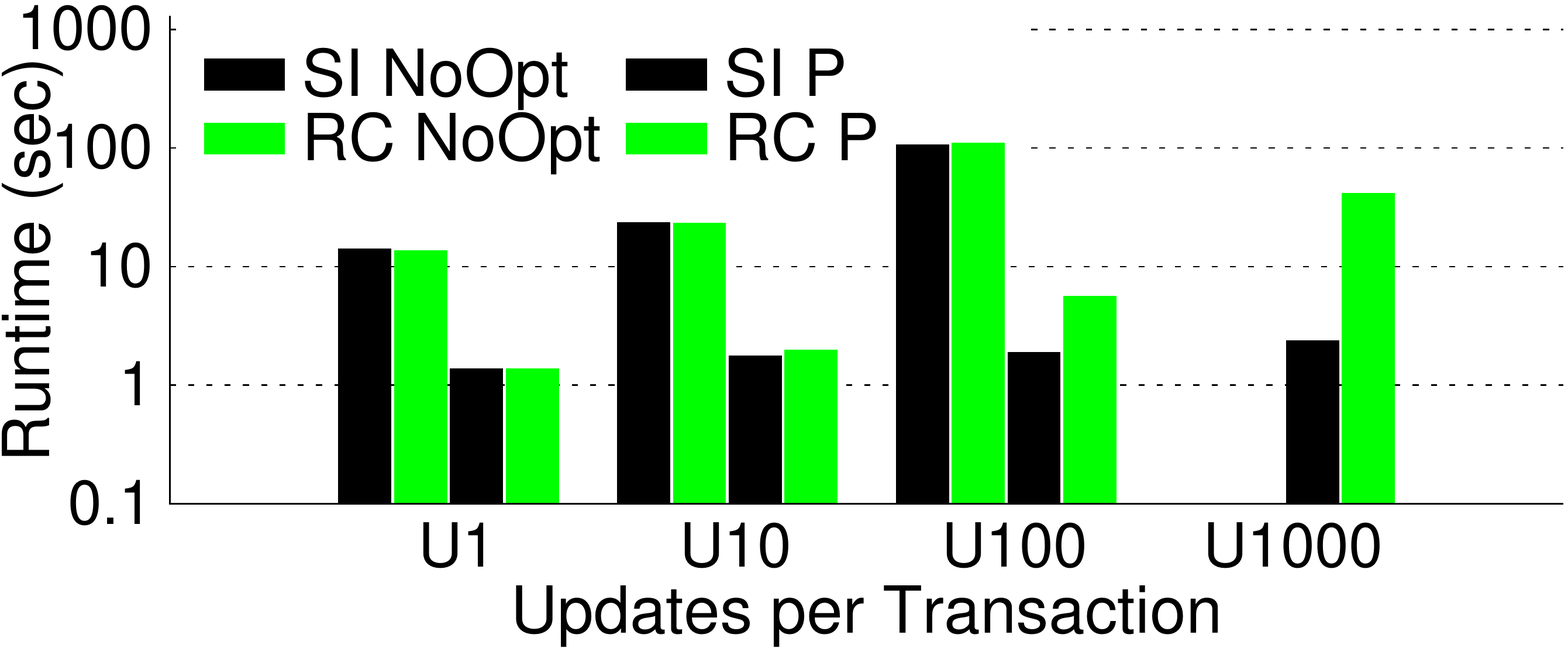}\\
  \vspace{-4mm}
  \caption{Isolation Levels}
  \label{fig:Different-Isolation-Level}
  \end{minipage}
  \begin{minipage}[b]{0.245\linewidth}
  \includegraphics[width=1\linewidth,trim=0 70pt 0 80pt, clip]{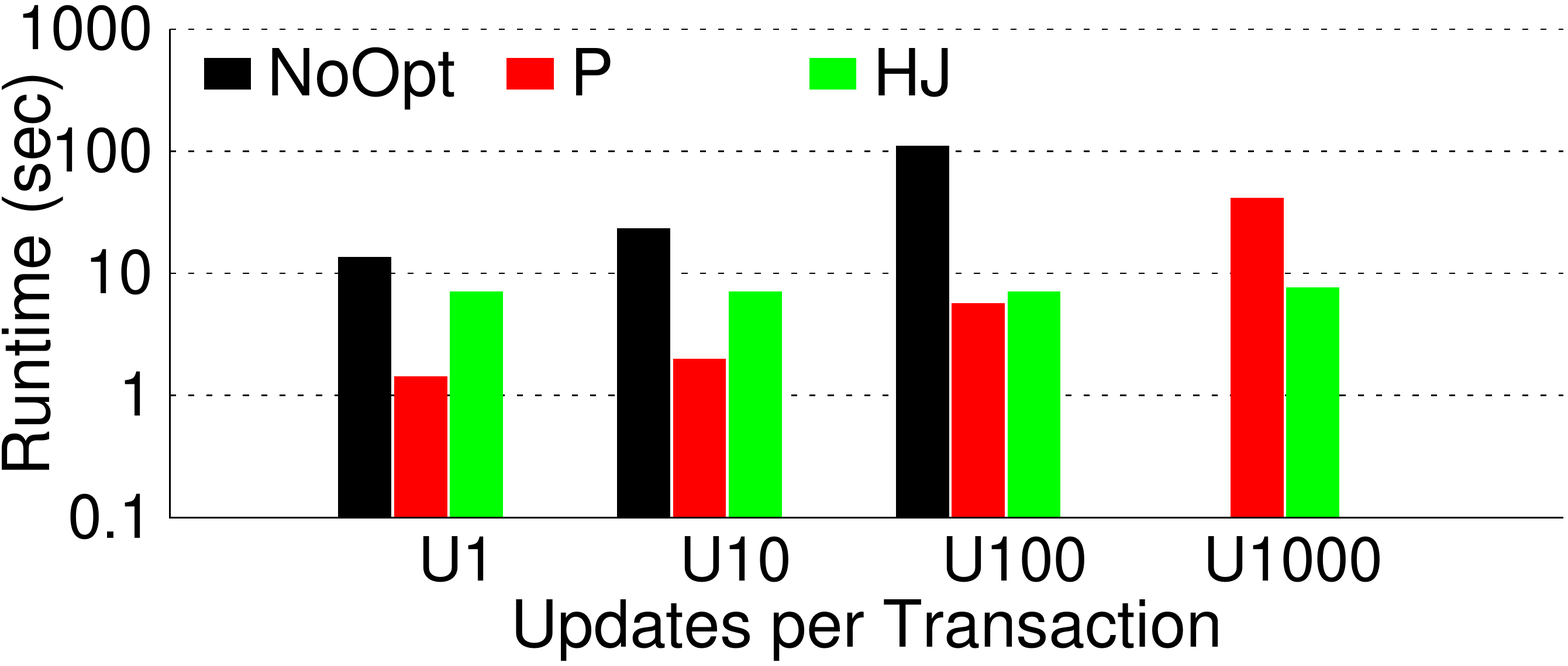}\\
  \vspace{-4mm}
  \caption{Optimization}
  \label{fig:Different-Optimization-Methods}
  \end{minipage}
  \begin{minipage}[b]{0.245\linewidth}
\includegraphics[width=1\linewidth,trim=0 70pt 0 150pt, clip]{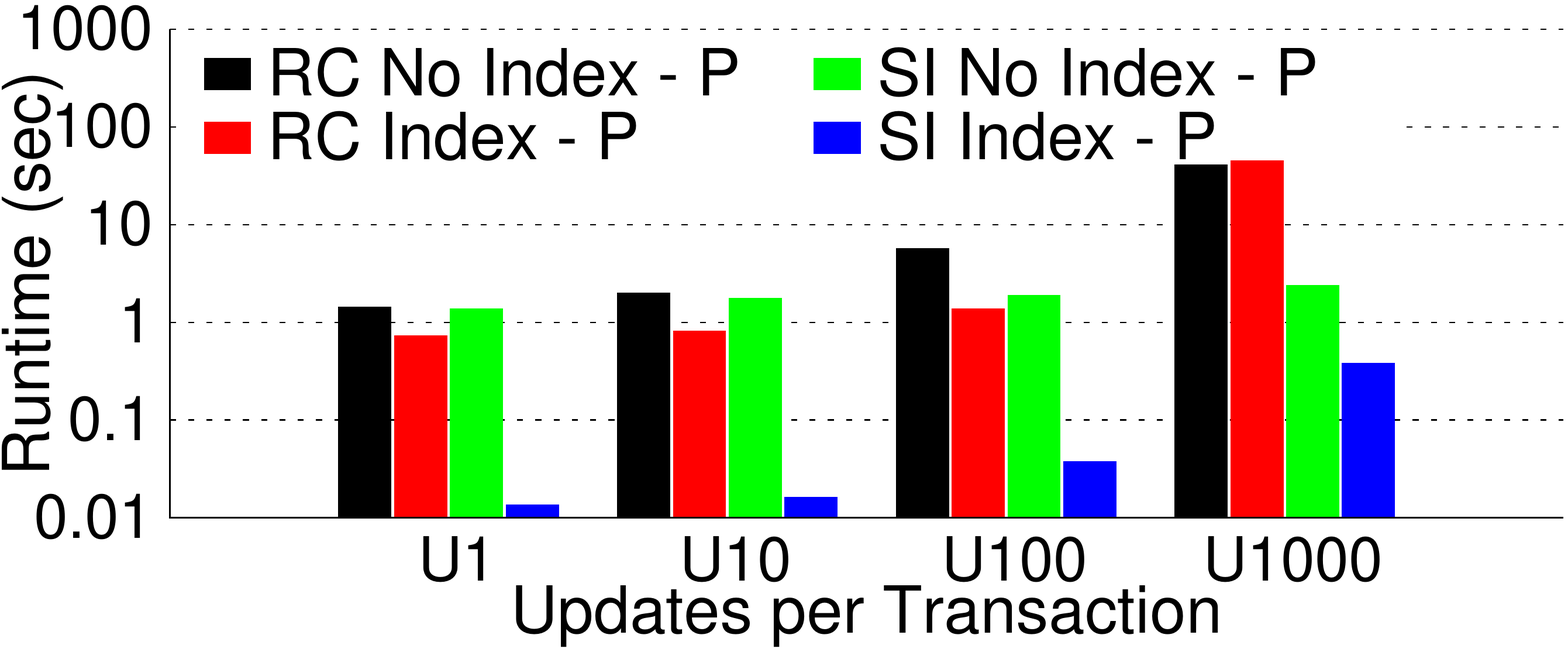}\\
  \vspace{-6.5mm}
  \caption{Index vs. No Index}
  \label{fig:Index-and-Isolation-Levels}
  \end{minipage}
  \begin{minipage}[b]{0.245\linewidth}
  \includegraphics[width=1\linewidth,trim=0 70pt 0 150pt, clip]{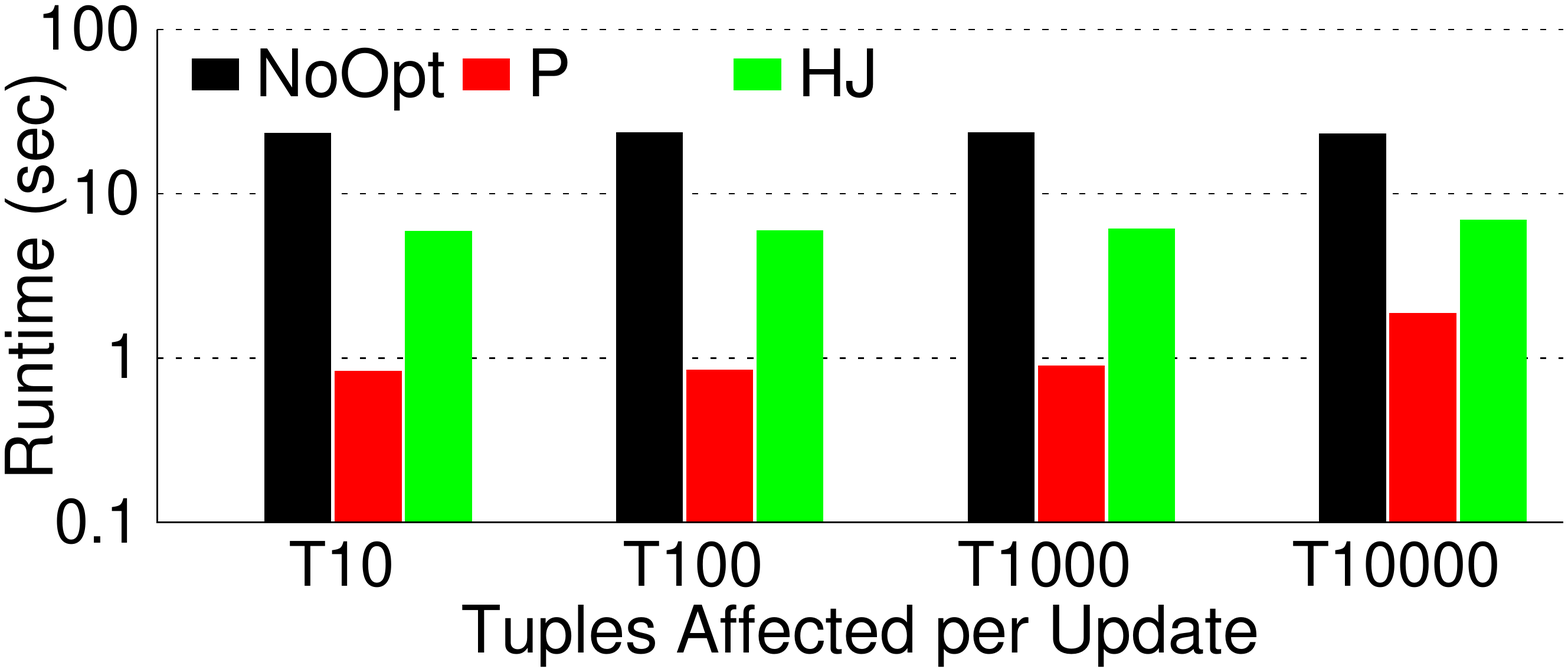}\\
  \vspace{-6.5mm}
  \caption{Affected Tuples}
  \label{fig:Number-of-Updated-Tuples}  
  \end{minipage}
  \begin{minipage}[b]{0.245\linewidth}
    \includegraphics[width=1\linewidth,trim=0 70pt 0 150pt, clip]{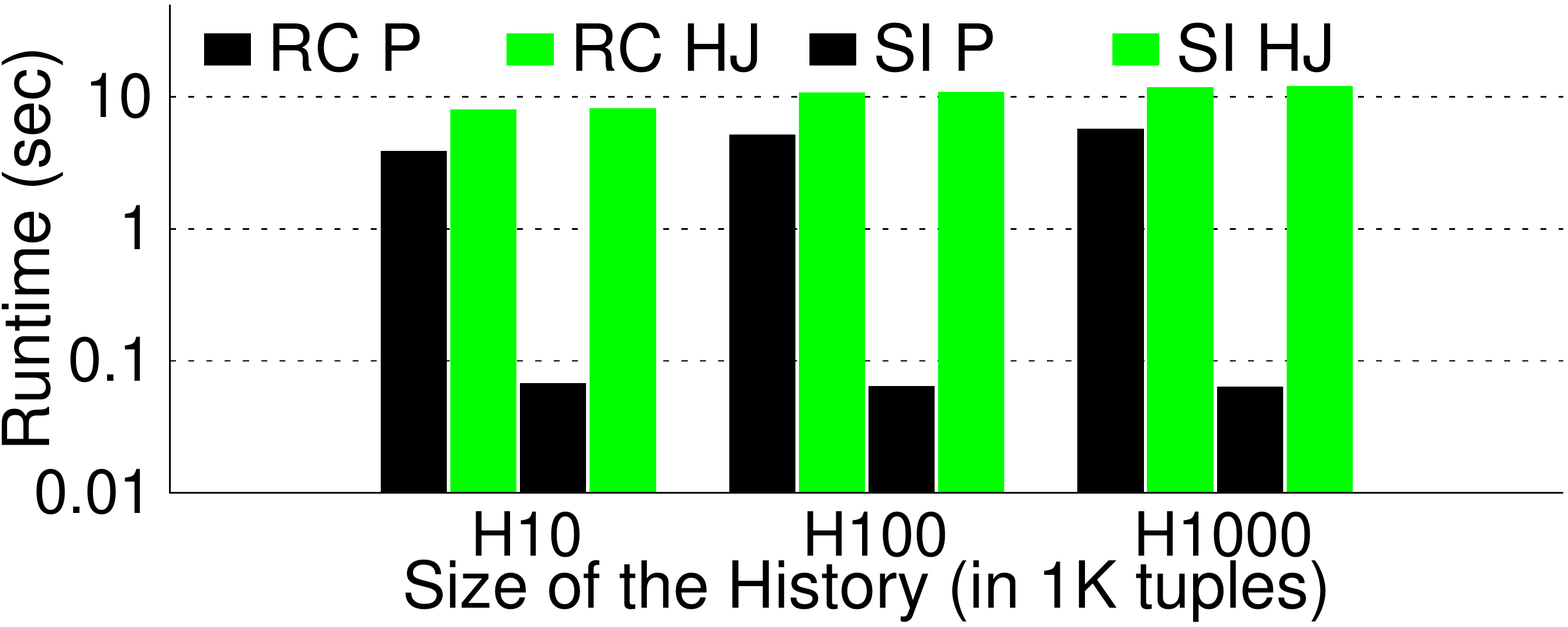}\\
  \vspace{-6.5mm}
  \caption{Inserts and Deletes}
  \label{fig:insert-delete-update}
  \end{minipage}
  \begin{minipage}[b]{0.245\linewidth}
  \includegraphics[width=1\linewidth,trim=0 70pt 0 150pt, clip]{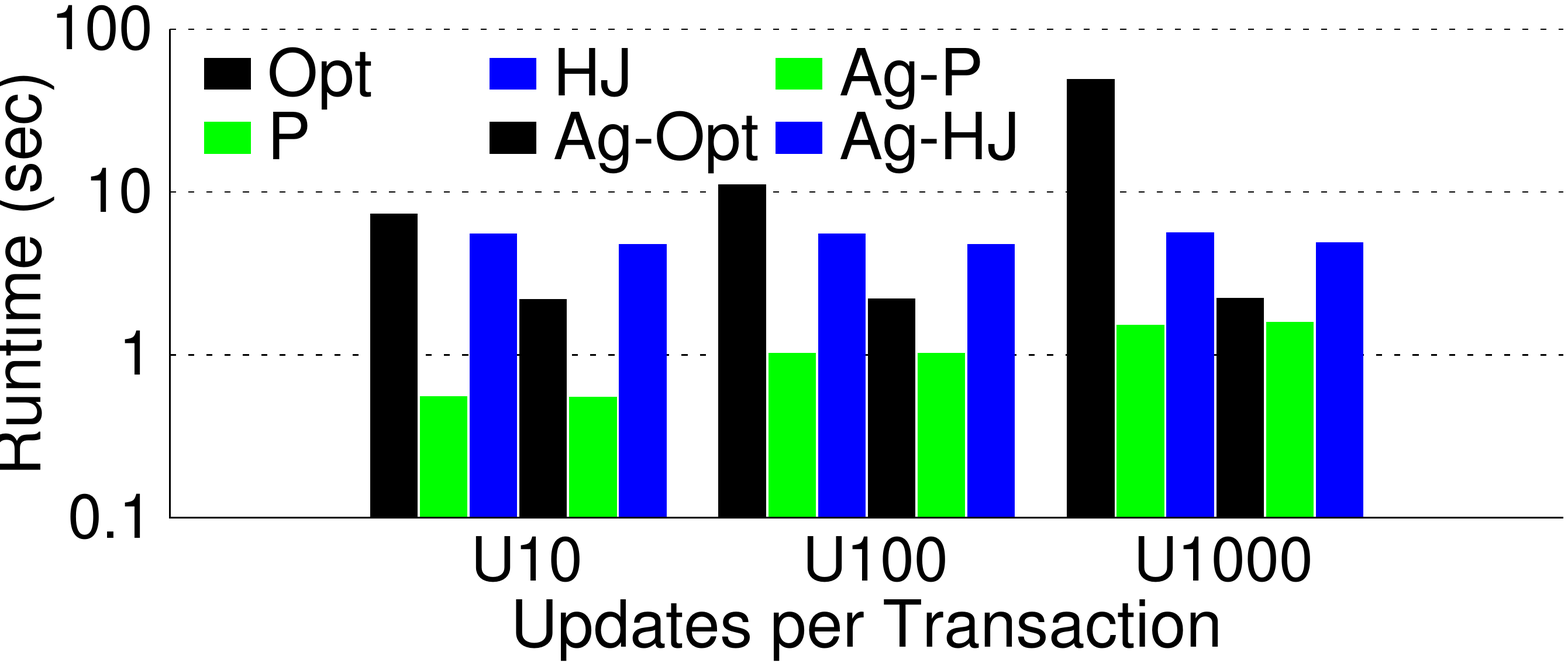}\\
  \vspace{-6.5mm}
  \caption{Aggregation}
  \label{fig:Agg-Prov-Info}
  \end{minipage}
\end{figure*}


\begin{figure}[t]
$,$\\[-4mm]
\begin{minipage}{1.0\linewidth}
  \begin{minipage}[b]{0.495\linewidth}
  \includegraphics[width=1\linewidth,trim=0 70pt 0 155pt, clip]{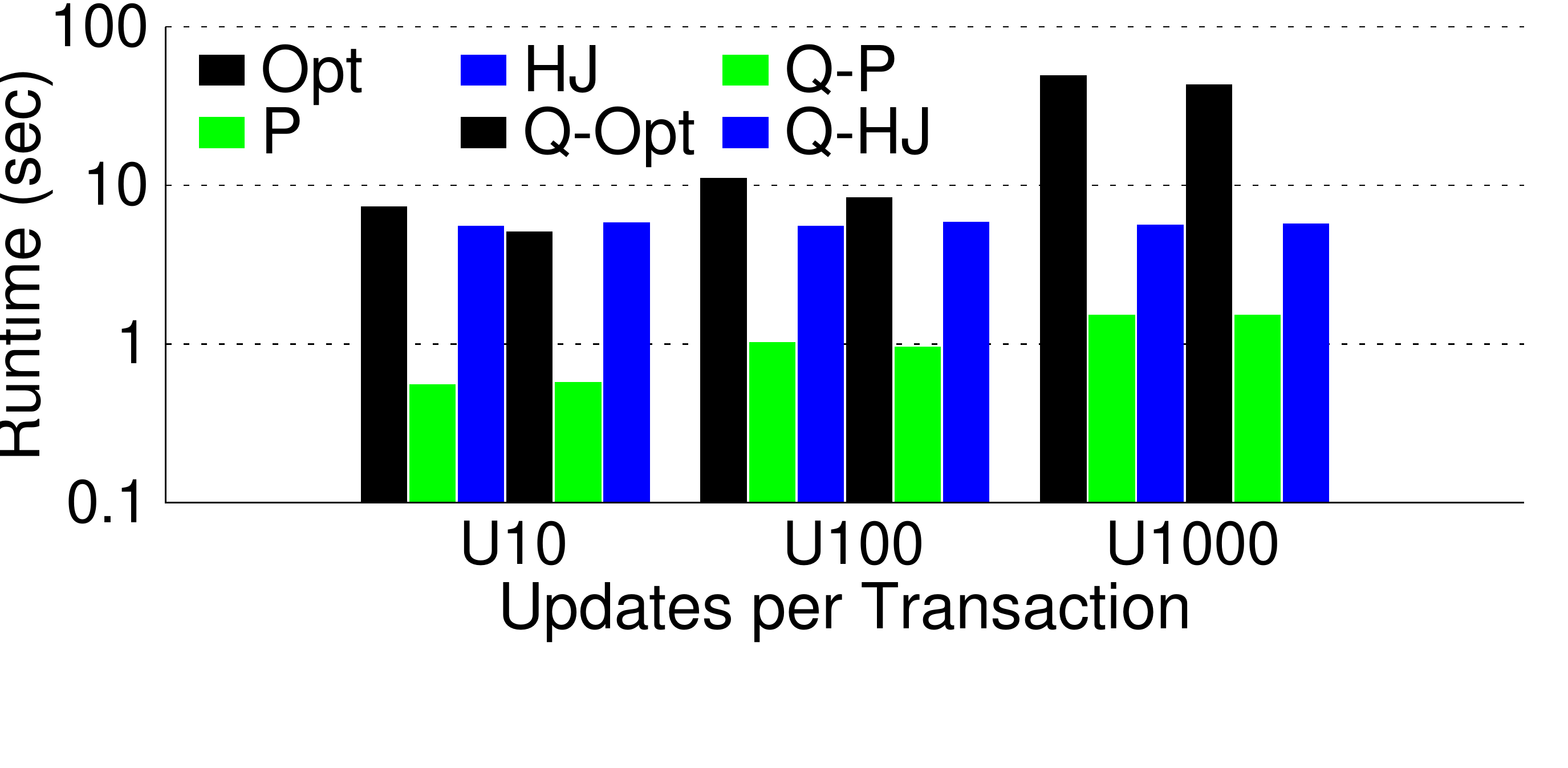}\\
  \vspace{-6.5mm}
  \caption{Query Provenance}
  \label{fig:Que-Prov-Comp}
  \end{minipage}
  \begin{minipage}[b]{0.495\linewidth}
\includegraphics[width=1\linewidth,trim=0 70pt 0 155pt, clip]{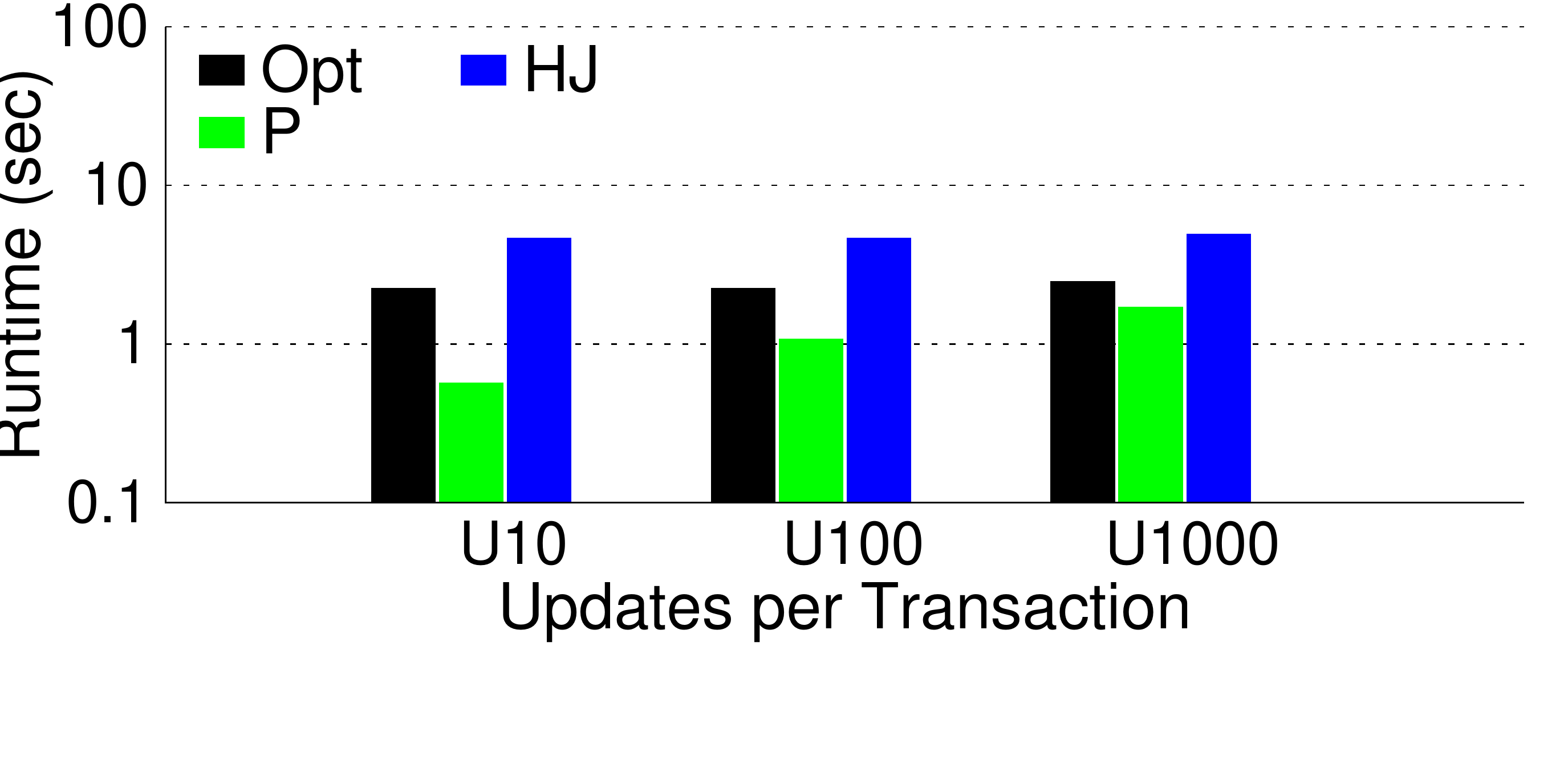}\\
  \vspace{-6.5mm}
  \caption{Query Vers. Ann.}
  \label{fig:Que-St-Anno}
  \end{minipage}
\end{minipage}
\end{figure}


\section{Experiments}
\label{sec:experiments}

Using commercial DBMS X, we evaluate 1) the performance of provenance computation using reenactment for isolation level \textit{RC-SI} and comparing it with \textit{SI}, and  2) the performance of querying provenance. 
All experiments were run on a machine with 2 x AMD Opteron 4238 CPUs (12 cores total), 128 GB RAM, and 4 x 1TB 7.2K
HDs in a hardware RAID 5 configuration. We have studied the runtime and storage overhead of DBMS X's build-in temporal and audit features in~\cite{AG16}. The results demonstrated that the runtime overhead for transaction execution is below 20\%  when audit logging and time travel are activated and it is more efficient than eager materialization of provenance during transaction execution (about 133\% overhead and higher). We did confirm the same trend for RC-SI and, thus, do not present these results here.



\noindent\textbf{Datasets and Workload}. In all experiments, we use a relation with five numeric
columns. Values for these attributes were generated randomly using a uniform distribution. 
Different variants $R10K$, $R100K$, and $R1000K$ with 10K, 100K, and 1M tuples and no significant history ($H0$) were created. Moreover, three variants  of $R1000K$ with different history sizes $H10$, $H100$, and $H1000$ (100K, 1M, and 10M tuples of history) are used. 
%
%
In most experiments, transactions consist only of update statements. The tuple
to be updated is chosen randomly by its primary key. 
The following parameters are used in experiments: $U$ is the number of updates per
transaction (e.g., \textit{U100} is a transaction with 100 updates). $T$ is the number of tuples affected by each update (default is $T1$). 
%
%
%
Transactions were executed under 
isolation level  \textit{RC-SI} (default) or
\textit{SI}.  
Experiments were repeated 100 times and the average runtime is reported.

\noindent\textbf{Compared Methods}. We apply different configurations for computing provenance of transactions using a subset of the optimizations outlined
in Section~\ref{sec:implementation}.
\textbf{NoOpt (N)}: Computes the provenance of all tuples in a relation including tuples that were not affected by the transaction. \textbf{Opt (O)}: Like the previous option but GProM's heuristic relational algebra optimizations are activated. 
\textbf{Prefilter (P)}: Only returns provenance of tuples affected by the transaction by prefiltering (Section~\ref{sec:implementation}).  
\textbf{HistJoin (HJ)}: Same as \textit{P}, but using the join method as described in Section~\ref{sec:implementation}. 
%

\partitle{Provenance Computation}
\label{sec:perf-prov-comp}
For the following experiments we have executed the transactional workload beforehand and measure performance of provenance capture. 

\noindent\textbf{Relation Size and Updates/Transaction}.
\label{sec:exp-vary-size-U}
We consider relations of different size ($R10K$, $R100K$, and $R1000K$) that do not have any significant history ($H0$).
Figure~\ref{fig:No-Significant-History} shows performance of computing provenance of transactions with different number of updates (\textit{U1} up to $U1000$). We applied \textit{N} and \textit{P}. We scale linearly in $R$ and $U$. By reducing the amount of data to be processed, 
the \textit{P} approach is orders of magnitude faster than the 
\textit{N} configuration. 

\noindent\textbf{History Size}.
\label{sec:exp-vary-hist}
%
Figure~\ref{fig:Compare-History-Size} shows the results for relations with 1M tuples ($R1000K$) and varying history sizes ($H0$, $H10$, $H100$, and $H1000$). We compute provenance of transactions with $10$ updates (\textit{U10}).  
Method \textit{N}  has almost constant performance for both isolation levels \textit{RC-SI} and \textit{SI}. 
The \textit{P} approach displays  better performance as it has to process less tuples. 
Its performance decreases for relations with a large history size.


\noindent\textbf{Isolation Levels}.
\label{sec:exp-vary-iso}
Figure~\ref{fig:Different-Isolation-Level} compares the result of transactions  under isolation levels \textit{SI} and \textit{RC-SI} with varying number of updates per transaction (\textit{U1} to \textit{U1000}). This experiment was conducted over table \textit{R1000K-H1000}. The runtime of \textit{N} is not
affected by the choice of isolation level, because the main difference between \textit{SI}
and \textit{RC-SI} reenactment is that 
we need to check whether
a row version is visibile for each update. 
However, the impact of these checks 
is
negligible for \textit{N} as the major cost factors are scanning the table and large
parts of its history as well as producing 1M output rows.
For the more efficient \textit{P} configuration this
effect is more noticeable, especially for larger number of updates per transaction. 
Note that for \textit{U1000} the \textit{N} method did not finish within the
 allocated time budget (1000 seconds). 


\noindent\textbf{Comparing Optimization Techniques}.
\label{sec:exp-comp-opti}
Figure~\ref{fig:Different-Optimization-Methods} compares different optimization methods (\textit{N}, \textit{P}, and \textit{HJ}) for varying number of updates (\textit{U1}, \textit{U10}, \textit{U100}, and \textit{U1000}) using \textit{R1000K-H1000}.  
Both \textit{P} and \textit{HJ} outperform \textit{N} with a more pronounced effect for larger number of updates per transaction.
\textit{P} outperforms \textit{HJ} for \textit{U1} by a factor of 5 whereas this result is reversed for \textit{U1000}. 
The runtime of \textit{HJ} is almost not
affected by parameter $U$, because it is dominated by the temporal join. 


\noindent\textbf{Index vs. No Index}.
\label{sec:exp-index}
We have studied the effect of using indexes for the relation storing the history of a relation. We use \textit{R1000K-H1000} and  vary $U$ (\textit{U1} to \textit{U1000}). 
Figure~\ref{fig:Index-and-Isolation-Levels} compares the effect of indexes for isolation levels \textit{RC-SI} and \textit{SI} using \textit{P}. The results demonstrate that using indexes improves execution time of queries that apply \textit{P} considerably. Provenance computation for \textit{SI} benefits more from indexes, because the prefilter conditions applied by the \textit{P} method are simpler for \textit{SI}.



\noindent\textbf{Affected Tuples Per Update}.
\label{sec:exp-tuples-per-update}
%
We now fix \textit{U10} and \textit{R1000K-H1000}, and vary the number of tuples  (\textit{T}) affected by each update from 10 to 10,000.
%
The runtime (Figure~\ref{fig:Number-of-Updated-Tuples}) 
 is dominated by scanning the
history and filtering out updated tuples (\textit{P}) or
the self-join between historic relations
(\textit{HJ}).
 Increasing the $T$ parameter by 3 orders of
magnitude increases runtime  by about 120\%
(\textit{P}) and 9\% (\textit{HJ}) whereas it does not effect runtime of queries using \textit{N}.



\noindent\textbf{Inserts and Deletes}.
\label{sec:exp-insert-delete}
We now consider transactions that use inserts, deletes, and updates over $R1000K$ varying history size ($H10$ to $H1000$). Each statement in a transaction is chosen randomly with equal probability to be an insert, update, or delete.
Figure~\ref{fig:insert-delete-update} presents the result for $U20$. Performance is comparable to performance for updates for \textit{RC-SI}. This aligns with our previous findings for \textit{SI}.

\partitle{Querying Provenance}
In GProM, provenance computations can be used as subqueries of a more complex SQL query. 
We now measure performance of querying provenance (the runtimes include the runtime of the subquery computing provenance). All experiments of this section are run over relation $R1000K-H0$ and transactions with \textit{U10} to \textit{U1000}.

\noindent\textbf{Aggregation of Provenance Information}.
\label{sec:Ag-Prov-Info}
Figure~\ref{fig:Agg-Prov-Info} shows the results for running an aggregation over the provenance computation 
(denoted as \textit{Ag-}). 
These results indicate that the performance of aggregation on provenance information is comparable to provenance computation. Even more, aggregation considerably improves performance for 
 (\textit{O}). For \textit{U1000}, \textit{Ag-O} results in 95\% improvement over \textit{O} (because it reduces the size of the output) while \textit{Ag-HJ} improves performance by $\sim$ 13\% compared to \textit{HJ}. 

\noindent\textbf{Filtering Provenance}.
\label{sec:Que-Prov-Comp}
A user may only be interested in part of the provenance that fulfills certain selection conditions, 
e.g., bonuses larger than a certain amount.
Figure~\ref{fig:Que-Prov-Comp} shows the runtime of provenance computation and querying (denoted as \textit{Q-}). Performance of querying the results of provenance capture is actually slightly better than just computing provenance, because it reduces the size of the output and selection conditions over provenance are pushed into the SQL query implementing the provenance computation. 

\noindent\textbf{Querying Versions Annotations}.
\label{sec:Que-Anno-Prov-Info}
A user can also query version annotations which are shown as boolean attributes in the provenance, e.g., to only return provenance for  tuples that were updated by a certain update of the transaction.
Figure~\ref{fig:Que-St-Anno} shows the performance results for such queries. We fix an update $u \in \xid$ and only return provenance of tuples modified by this update. 
This reduces the runtime of \textit{O} queries significantly by reducing  the size of the output.

\section{Conclusions}
\label{sec:conclusions}

We have presented an efficient solution for computing the provenance of
transactions run under RC-SI by extending our MV-semiring model and reenactment approach.
Our experimental
evaluation demonstrates that our novel optimizations specific to RC-SI enables us to achieve performance comparable to SI reenactment.
In future work, we would like to explore the application of reenactment for post-mortem debugging of transactions which is particularly important for lower isolations level such as RC-SI.







\bibliographystyle{abbrv}
\bibliography{trans}

\clearpage
\appendix

\begin{figure*}[t]
  \centering
 \begin{align*}
         \isMax(R,k) &= \begin{cases}
      0 & \mathtext{if} \exists t',j : \idOf(R(t')[j]) = \idOf(k) \wedge \versionOf(R(t')[j]) > \versionOf(k) \\
      1 & \mathtext{otherwise}
      \end{cases}\\
      \isStrictMax(R,k) &= \begin{cases}
      0 & \mathtext{if} \exists t',j : \idOf(R(t')[j]) = \idOf(k) \wedge \versionOf(R(t')[j]) \geq \versionOf(k) \\
      1 & \mathtext{else}
      \end{cases}\\
         \idOf(\upMarker{X}{\xid}{\version}{\tid}(k')) &=  \tid \\
      \versionOf(\upMarker{X}{\xid}{\version}{\tid}(k')) &= \version 
 \end{align*}  
  \caption{Definitions of $\isMax$, $\isStrictMax$, $idOf$ and $\versionOf$.}
  \label{fig:redef-ismax}
\end{figure*}

\section{Proofs}
\label{sec:proofs}

Before proving Theorem~\ref{theo:update-reenactment-equivalence-RCSI} we establish that lifted homomorphisms commute with RC-SI histories and the two new query operators we have introduced for RC-SI reenactment. Note that~\cite{AG16} established that lifted homomorphisms commute with queries, updates, and SI histories.
Lifted homomorphisms are a special type of $\mvK$-homomorphisms that are derived from a $\semK$-homomorphism $h$ by applying it to all elements $k \in K$ that occur an $\mvK$-element $k^{\version}$. That is, a lifted homomorphism preserves the expression structure of MV-semiring elements.  For example, consider a homomorphism $h: \ppSR \to \mathbb{N}$ defined as
\begin{align*}
  h(k) =
  \begin{cases}
    1 & \mathtext{if} k = x_1\\
    3 & \mathtext{if} k = x_2\\
    0 & \mathtext{otherwise}\\
  \end{cases}
\end{align*}

Applying the corresponding lifted homomorphism $\liftH{h}$ to 

$$\upMarker{C}{\xid}{3}{1}(\upMarker{U}{\xid}{1}{1}(x_1)) + \upMarker{C}{\xid}{3}{2}(\upMarker{U}{\xid}{2}{2}(x_2))$$

 yields

$$\upMarker{C}{\xid}{3}{1}(\upMarker{U}{\xid}{1}{1}(1)) + \upMarker{C}{\xid}{3}{2}(\upMarker{U}{\xid}{2}{2}(3))$$.

\begin{lem}\label{lem:lift-h-commutes-RC-SI-histories}
Let $\semK_1$ and $\semK_2$ be commutative semirings and $h: \semK_1 \to \semK_2$ a semiring homomorphism. Then the lifted homomorphism $\liftH{h}: \mvOf{\semK_1} \to \mvOf{\semK_2}$ as defined in~\cite{AG16} commutes with any RC-SI history $\history$.
\end{lem}
\begin{proof}
As mentioned above and proven in~\cite{AG16}, $\liftH{h}$ commutes with queries, updates, and SI histories. In the definition of $\relV{R}{\xid}{\version}$, the committed relation version $\relCV{R}{\version}$ is defined analog to SI histories. The same is true for predicate $\hasUp(T,$ $t,k,\version)$ and $\validAt(T,t,k,\version)$. Based on Theorem 5.5 of ~\cite{AG16} any lifted homomorphism commutes with $\hasUp(T,t,k,\version)$ and $\validAt(T,t,k,\version)$ as well as with the operations used in the definition of $\relCV{R}{\version}$. Since these results do not depend on the admissibility of the input relation (which is based on the concurrency control protocol and thus different for SI and RC-SI), it only remains to show that the lifted homomorphism $\liftH{h}$ commutes with the operations of $\relVE{\rel}{\xid}{\version}(t)$, i.e., it can be pushed into the committed relation version accessed by $\relVE{\rel}{\xid}{\version}(t)$. We have

\begin{align*}
&\liftH{h}(\relVE{\rel}{\xid}{\version})(t) \\
=&  \liftH{h}(\sum_{i=0}^{\numInSum{\relCV{\rel}{\version}(t)}} \nthOfK{\relCV{\rel}{\version}(t)}{i} \times \validEx (\xid,t, \nthOfK{\relCV{\rel}{\version}(t)}{i}, \version)\\
&+ 
\sum_{i=0}^{\numInSum{\relV{\rel}{\xid}{\version-1}(t)}} \nthOfK{\relV{\rel}{\xid}{\version - 1}(t)}{i}\\ &\mathtab\mathtab\mathtab\mathtab\mathtab\mathtab\times \validIn (\xid,t, \nthOfK{\relV{\rel}{\xid}{\version - 1}(t)}{i}, \version - 1))
\end{align*}
Any homomorphism $\liftH{h}$ commutes with addition. Thus,
\begin{align*}
=&  \sum_{i=0}^{\numInSum{\liftH{h}(\relCV{\rel}{\version}(t))}} \liftH{h}(\nthOfK{\relCV{\rel}{\version}(t)}{i} \times \validEx (\xid,t, \nthOfK{\relCV{\rel}{\version}(t)}{i}, \version))\\
&+ 
\sum_{i=0}^{\numInSum{\liftH{h}(\relV{\rel}{\xid}{\version-1}(t))}} \liftH{h}(\nthOfK{\relV{\rel}{\xid}{\version - 1}(t)}{i}\\
& \mathtab\mathtab\mathtab\mathtab\mathtab\mathtab \times \validIn (\xid,t, \nthOfK{\relV{\rel}{\xid}{\version - 1}(t)}{i}, \version - 1)))
\end{align*}

and since $\liftH{h}$ also commutes with multiplication, we have

\begin{align*}
=&  \sum_{i=0}^{\numInSum{\liftH{h}(\relCV{\rel}{\version}(t)})} \liftH{h}(\nthOfK{\relCV{\rel}{\version}(t)}{i}) \times \liftH{h}(\validEx (\xid,t, \nthOfK{\relCV{\rel}{\version}(t)}{i}, \version))\\
&+ 
\sum_{i=0}^{\numInSum{\liftH{h}(\relV{\rel}{\xid}{\version-1}(t)})} \liftH{h}(\nthOfK{\relV{\rel}{\xid}{\version - 1}(t)}{i}) \\
& \mathtab\mathtab\mathtab\mathtab\mathtab\mathtab\times \liftH{h}(\validIn (\xid,t, \nthOfK{\relV{\rel}{\xid}{\version - 1}(t)}{i}, \version - 1))
\end{align*}

Given that $\liftH{h}(\hasUp(T,t,k,\version)) = \hasUp(T,t,\liftH{h}(k),\version)$, it follows that 

\begin{align*}
&\liftH{h}(\validEx (\xid,t, \nthOfK{\relCV{\rel}{\version}(t)}{i}, \version))\\
= &\validEx (\xid,t, \nthOfK{\liftH{h}(\relCV{\rel}{\version})(t)}{i}, \version)
\end{align*}

Furthermore, the condition $\validIn$ is based only on the outermost version annotation in a summand $k$. Since lifted homomorphisms by design do not manipulate version annotations it follows that:

\begin{align*}
&\liftH{h}(\validIn (\xid,t, \nthOfK{\relV{\rel}{\xid}{\version - 1}(t)}{i}, \version - 1))\\
= &\validIn (\xid,t, \nthOfK{\liftH{h}(\relV{\rel}{\xid}{\version - 1})(t)}{i}, \version - 1)
\end{align*}

Thus, we have

\begin{align*}
=&  \sum_{i=0}^{\numInSum{\liftH{h}(\relCV{\rel}{\version})(t)}} \nthOfK{\liftH{h}(\relCV{\rel}{\version})(t)}{i} \times \validEx (\xid,t, \nthOfK{\liftH{h}(\relCV{\rel}{\version}(t))}{i}, \version)\\
&+ 
\sum_{i=0}^{\numInSum{\liftH{h}(\relV{\rel}{\xid}{\version-1})(t)}} 
\nthOfK{\liftH{h}(\relV{\rel}{\xid}{\version - 1})(t)}{i}\\
& \mathtab\mathtab\mathtab\mathtab\mathtab\mathtab \times \validIn (\xid,t, \nthOfK{\liftH{h}(\relV{\rel}{\xid}{\version - 1})(t)}{i}, \version - 1)
\end{align*}

This implies that $\liftH{h}$ can be pushed into $\relVE{\rel}{\xid}{\version}$ and given that $\liftH{h}$ commutes with all other operations used to define $\relV{\rel}{\xid}{\version}$ it follows that $\liftH{h}$ commutes with histories. 
\end{proof}

Furthermore, we have introduced two new query operators that are used in reenactment. We now prove that lifted homomorphisms commute with these query operators. This means we only need to prove $\ppSRV$-equivalence of operations with their reenactment queries, because this then automatically implies $\mvK$-equivalence for any naturally ordered semiring $\mvK$. The new query  operators we have introduced are the \textit{version merge} operator $\vMerge(R,S)$ that merges two versions $R$ and $S$ of the same relation by only keeping the newest versions of tuples and the \textit{version filter} operator $\vFilt{\theta}(R)$ which removes summands (tuple versions) which do not fulfill the condition $\theta$ expressed over the versions (pseudo attribute $V$) encoded in the version annotations.

\begin{lem}\label{lem:lifted-commutes-version-merge}
Let $\liftH{h}: \mvOf{\semK_1} \to \mvOf{\semK_2}$ be a lifted homomorphism, then $\liftH{h}$ commutes with $\vMerge(R,S)$ if $R$ and $S$ are normalized admissible $\mvOf{\semK_1}$-relations.
\end{lem}
\begin{proof}

\begin{align*}
&\liftH{h}(\vMerge(R,S))(t)\\ 
= &\liftH{h}(\sum_{i=0}^{\numInSum{R(t)}} R(t)[i] \times \isMax(S,R(t)[i]) \\ 
                &+\sum_{i=0}^{\numInSum{S(t)}} S(t)[i] \times \isStrictMax(R,S(t)[i]))\\
\end{align*}

Any homomorphism commutes with addition and multiplication. Furthermore, since $\liftH{h}$ preserves the structure of MV-semiring expressions, we know that $\liftH{h}(k^{\version})$ for any normalized MV-semiring element $k^{\version}$ is a subset of the summands of $k^{\version}$. That is every summand in $k^{\version}$ is preserved unless  $h(k) = 0$ for all elements $k \in K$ that occur in the summand, because in this case the summand's expression is equivalent to $0$ resulting in the summand being removed. Thus, as long as we can prove that $\isMax(R,k) = \isMax(R,$ $\liftH{h}(k))$ and $\isStrictMax(R,k) = \isStrictMax(R,\liftH{h}(k))$ it follows that:

\begin{align*}
 = &\sum_{i=0}^{\numInSum{\liftH{h}(R(t))}} \liftH{h}(R(t)[i] \times \isMax(S,R(t)[i])) \\ 
                &+\sum_{i=0}^{\numInSum{\liftH{h}(S(t))}} \liftH{h}(S(t)[i] \times \isStrictMax(R,S(t)[i]))\\
\end{align*}

Consider the definition of $\isMax$, $\isStrictMax$, $idOf$ and $\versionOf$ as shown in Figure~\ref{fig:redef-ismax}.
Note that $\liftH{h}(\idOf(k)) = \idOf(k)$ and $\liftH{h}(\versionOf(k)) = \versionOf(k)$, because by construction of $\liftH{h}$ we have $\liftH{h}(\upMarker{X}{\xid}{\version}{\tid}(k')) = \upMarker{X}{\xid}{\version}{\tid}(\liftH{h}(k'))$ and thus $\liftH{h}(\idOf(\upMarker{X}{\xid}{\version}{\tid}(k'))) = \idOf(\upMarker{X}{\xid}{\version}{\tid}(k'))$ as well as $\liftH{h}(\versionOf(\upMarker{X}{\xid}{\version}{\tid}(k'))) = \versionOf(\upMarker{X}{\xid}{\version}{\tid}(k'))$.
From this immediately follows that $\isMax(R,k) = \isMax(R,\liftH{h}(k))$ and $\isStrictMax(R,k) = \isStrictMax(R,\liftH{h}(k))$ which concludes the proof.
\end{proof}

\begin{lem}\label{lem:lifted-commutes-version-filter}
Let $\liftH{h}: \mvOf{\semK_1} \to \mvOf{\semK_2}$ be a lifted homomorphism, then $\liftH{h}$ commutes with $\vFilt{\theta}(R,S)$ if $R$ and $S$ are normalized admissible $\mvOf{\semK_1}$-relations.
\end{lem}
\begin{proof}
Substituting the definition of $\vFilt{\theta}$ we get:

\begin{align*}
&\liftH{h}( \vFilt{\theta}(R))(t) \\
= &\liftH{h}(\sum_{i=0}^{\numInSum{R(t)}} R(t)[i] \times \theta(R(t)[i]))\\
= &\sum_{i=0}^{\numInSum{\liftH{h}(R(t))}} \liftH{h}(R(t)[i]) \times \liftH{h}(\theta(R(t)[i]))\\
\end{align*}

Recall that $\theta(k)$ is evaluated over the version $\version$ of the outermost version annotation of each summand $k_i$ in the normalized annotation $k$. Thus, we get 

\begin{align*}
= &\sum_{i=0}^{\numInSum{\liftH{h}(R)(t)}} \liftH{h}(R)(t)[i] \times \theta(\liftH{h}(R)(t)[i])\\
=  &\vFilt{\theta}(\liftH{h}(R)))(t)
\end{align*}
\end{proof}

Finally, the following lemma establishes that if a Transaction $\xid$ uses only updates, deletes, and inserts with singleton relations (operator $\asingleton{t}{k}$ corresponding to an SQL statement of the form \lstinline!INSERT INTO ... VALUES ...)! then $\relCV{R}{\finish{\xid} - 1}$ contains all immediate predecessors of all tuple versions created by $\xid$ 's updates and deletes. This is the first prerequisite for proving the correctness of $\ract_{opt}(\xid)$, because $\ract_{opt}$ avoids the use of the version merge operator by only using accesses to relation versions as of $\finish{\xid}-1$. In the following definition we make use of a predicate $\hasCreated(\xid,t,k)$ which determines whether a summand $k$ in the annotation of a tuple $t$ has been created by Transaction $\xid$. Formally,

$$
\hasCreated(\xid,t,k) \Leftrightarrow \exists k': k = \upMarker{X}{\xid}{\version}{\tid}(k')
$$

\begin{defi}
Let $\history$ be a RC-SI history and $\xid \in \history$. Consider a summand $k$ in the annotation $\relV{R}{\xid}{\finish{\xid}}(t)$ created by $\xid$ ($\hasCreated(\xid,t,k)$ is true). The immediate predecessor $\impred(\xid,t,k)$ is defined as the latest tuple version $k'$ with identifier $\idOf(k') = \idOf(k)$  created by a transaction $\xid' \neq \xid$ in the annotation of a tuple $t'$. If no such version exists (e.g., $\xid$ did insert $k$) then $\impred(\xid,t,k)$ is undefined.
\end{defi}

In other words, the immediate predecessor of a tuple version $k$ is the last version of this tuple created by another transaction before the creation of $k$.

\begin{lem}\label{lem:end-minus-one-has-all-pred}
Let $\xid$ be a transaction where each insert's query is of the form $\asingleton{t}{k}$. If $\xid$ is executed as part of a RC-SI history $\history$ then there exists a tuple $t'$ such that $\impred(\xid,t,k)$ is present in $\relCV{R}{\finish{\xid}-1}(t')$. 
\end{lem}
\begin{proof}
For any tuple version $k$ created by Transaction $\xid$, there has to exist an operation $u_i$ in $\xid$ that first created a tuple version $k'$ with $\idOf(k) = \idOf(k')$. Naturally, $\relCV{R}{\version(u_i)(t')}$ for some tuple $t'$ has to contain $\impred(\xid,t,k)$ if it is defined. We proof the lemma by contradiction. Assume that $\relCV{R}{\finish{\xid} - 1}(t')$ does not contain $\impred(\xid,t,k)$. This can only be the case if there exists a Transaction $\xid''$ with $\finish{\xid''} < \finish{\xid}$ that did update or delete $k$. However, since $\up_i$ modified $k$ we know that $\xid$ would have to hold a write lock on the tuple version corresponding to $k$ after $\version(u_i)$ and under RC-SI write locks are held until transaction commit. Thus, no such Transaction $\xid''$ can exist. 
\end{proof}

\begin{theoremproof}{\bf \ref{theo:update-reenactment-equivalence-RCSI}}
  Let $\xid$ be a RC-SI transaction. 
  Then, 

$$\xid \equiv_{\ppSRV} \ract(\xid) \equiv_{\ppSRV} \ract_{opt}(\xid)$$
\end{theoremproof}

\begin{proof}
We first prove that $\xid \equiv_{\ppSRV} \ract(\xid)$ and then equivalence with $\ract_{opt}$.\\

\prooftitle{$\xid \equiv_{\ppSRV} \ract(\xid)$}

Assume that transaction $\xid = u_1, \ldots, u_n, c$ is updating a single relation $R$. As was shown in~\cite{AG16}, the extension to multiple relations is straightforward. To prove equivalence it suffices to show that a reenactment query for an update $\ract(\up)$ is equivalent to the update $\up$ and that each such reenactment query is executed over the same input relation as in the original history $\history$. The semantics for updates is the same under SI and RC-SI. The proof of $\up \equiv_{\ppSRV} \ract(\up)$ was already given in~\cite{AG16}. It remains to show that the input $\relV{R}{\xid}{\version(\up)}$ is the same as the input produced for $\ract(\up)$ by the reenactment query for Transaction $\xid$.

We prove this fact by induction over the number of updates in Transaction $\xid$. 

\prooftitle{Induction Start} Let $\xid = u_1, c$. This case is analog to SI and thus was already proven in~\cite{AG16}.

\prooftitle{Induction Step} Assume that $\relV{R}{\xid}{\version(\up_i)} = \relVE{R}{\xid}{\version(\up_i)}$ with $i \in \{ 1, \ldots, m\}$ where $m$ is the number of operations in the Transaction $\xid$ is correctly constructed by the reenactment query for $\xid$ for any transaction with  $m < n$ operations. We need to prove that for any transaction $\xid = u_1,\ldots,u_{n+1}, c$ we have that $\relV{R}{\xid}{\version(\up_{n+1})}$ is equal to the input for the reenactment query $\ract(\up_{n+1})$ of $\up_{n+1}$ within the reenactment query $\ract(\xid)$. In the reenactment query, the input to $\ract(\up_{n+1})$ is $\vMerge(\ract(u_{n}),\relCV{R}{\version(u_{n+1})})$.

\begin{align*}
  &\vMerge(\ract(u_{n}),\relCV{R}{\version(u_{n+1})})(t)\\
= & \sum_{i=0}^{\numInSum{\ract(u_{n})(t)}} \ract(u_{n})(t)[i] \times \isMax(\relCV{R}{\version(u_{n+1})},\ract(u_{n})(t)[i]) \\ 
&+\sum_{i=0}^{\numInSum{\relCV{R}{\version(u_{n+1})}(t)}} \relCV{R}{\version(u_{n+1})}(t)[i] \\
&\mathbigtab\mathbigtab\times \isStrictMax(\ract(u_{n}),\relCV{R}{\version(u_{n+1})}(t)[i])
\end{align*}

Based on the induction hypothesis we have 

$$\ract(u_{n}) = \relV{R}{\xid}{\version(\up_{n+1})}$$. 

Thus, denoting $\version(\up_{n+1})$ as $\version_{n+1}$:

\begin{align*}
= & \sum_{i=0}^{\numInSum{\relV{R}{\xid}{\version_{n+1}}(t)}} \relV{R}{\xid}{\version_{n+1}}(t)[i] \\
&\mathbigtab\mathbigtab\times \isMax(\relCV{R}{\version_{n+1}},\relV{R}{\xid}{\version_{n+1}}(t)[i]) \\ 
&+\sum_{i=0}^{\numInSum{\relCV{R}{\version_{n+1}}(t)}} \relCV{R}{\version_{n+1}}(t)[i] \\
&\mathbigtab\mathbigtab\times \isStrictMax(\relV{R}{\xid}{\version_{n+1}},\relCV{R}{\version_{n+1}}(t)[i])
\end{align*}

Note that $\relVE{R}{\xid}{\version_{n+1}}(t)$ is also defined as a sum over the elements from  $\relV{R}{\xid}{\version_{n+1}}(t)$ and $\relCV{R}{\version_{n+1}}(t)$. Individual summands are  filtered out using $\validIn$ and $\validEx$. Thus, to proof that $\vMerge(\ract(u_{n}),\relCV{R}{\version(u_{n+1})}) = \relV{R}{\xid}{\version(\up_{n+1})}$, we have to show that if either the $\isMax$ or $\isStrictMax$ function returns $1$ on a summand then the same is true for $\validIn$ respective $\validEx$ and vice versa. 

Fixing a tuple $t$, we have to distinguish between five cases for each tuple version (summand) $k$ in the annotation of tuple $t$ as shown below. Table~\ref{tab:cases-of-tuple-version-spread} shows the versions of a tuple version with an identifier $\tid$ in $\relV{R}{\xid}{\version_{n+1}}$ and $\relCV{R}{\version_{n+1}}$ for each of the cases.
 \begin{enumerate}
\item $k$ is the latest version of all tuple versions with identifier $\idOf(k)$ and was created by Transaction $\xid$ before $\version_{n+1}$. In this case $k$ is only present in $\relV{R}{\xid}{\version_{n+1}}(t)$. For this case we assume that the first tuple version with identifier $\idOf(k)$ was created by an insert of Transaction $\xid$. Thus, there cannot exist an outdated version $k'$ with this identifier in the annotation of any tuple $t'$ in $\relCV{R}{\version_{n+1}}$. 
\item $k$ is the latest version of all tuple versions with identifier $\idOf(k)$ and was created by a Transaction $\xid$ before $\version_{n+1}$. In this case $k$ is only present in $\relV{R}{\xid}{\version_{n+1}}(t)$. The previous tuple version with identifier $\idOf(k)$ was created by a Transaction $\xid' \neq \xid$. Hence, there has to exist an outdated version $k'$ with this identifier in the annotation of some tuple $t'$ in $\relCV{R}{\version_{n+1}}$. 
\item $k$ is the latest version of all tuple versions with identifier $\idOf(k)$ and was created by a Transaction $\xid'$ that committed after $\start{\xid}$, but before $\version_{n+1}$. In this case $k$ is only present in $\relCV{R}{\version_{n+1}}(t)$. 
For this case we assume that the previous tuple version with identifier $\idOf(k)$ was created by an insert of Transaction $\xid' \neq \xid$. Thus, there cannot exist an outdated version $k'$ with this identifier in the annotation of any tuple $t'$ in $\relV{R}{\xid}{\version_{n+1}}$. 
\item $k$ is the latest version of all tuple versions with identifier $\idOf(k)$ and was created by a Transaction $\xid'$ that committed after $\start{\xid}$, but before $\version_{n+1}$. In this case $k$ is only present in $\relCV{R}{\version_{n+1}}(t)$. 
The first tuple version with identifier $\idOf(k)$ was created by an insert of a Transaction $\xid'' \neq \xid$ where $\finish{\xid''} < \version_{n+1}$. Hence, there has to exist an outdated version $k'$ with this identifier in the annotation of some tuple $t'$ in $\relV{R}{\xid}{\version_{n+1}}$.
 \item $k$ is the latest version of all tuple versions with identifier $\idOf(k)$ and was created by a Transaction $\xid'$ that committed before $\start{\xid}$. In this case $k$ is present in both $\relV{R}{\xid}{\version_{n+1}}(t)$ and $\relCV{R}{\version_{n+1}}(t)$.
 \end{enumerate}

 \begin{table*}[t]
   \centering
   \textbf{Occurrence of summand with identifier $\tid$}\\[2mm]
   \begin{tabular}{|l|c|c|}
     \hline
     \thead{Case}&\thead{$\relCV{R}{\version_{n+1}}$}&\thead{$\relV{R}{\xid}{\version_{n+1}}$}\\
     1 & none present & latest version $k$ with $\idOf(k) = \tid$ \\ 
     2 & outdated version $k'$ with $\idOf(k) = \tid$ & latest version $k$ with $\idOf(k) = \tid$ \\ 
     3 &  latest version $k$ with $\idOf(k) = \tid$ & none present \\ 
     4 &  latest version $k$ with $\idOf(k) = \tid$ & outdated version $k'$ with $\idOf(k) = \tid$ \\ 
     5 &  latest version $k$ with $\idOf(k) = \tid$ &  latest version $k$ with $\idOf(k) = \tid$\\
     \hline
   \end{tabular}
   \caption{Cases of how tuple versions with a fixed identifier $\tid$ can occur in $\relCV{R}{\version_{n+1}}$ and $\relV{R}{\xid}{\version_{n+1}}$.}
   \label{tab:cases-of-tuple-version-spread}
 \end{table*}

\prooftitle{Case 1}
Since $k$ is the only summand with identifier $\tid$ in $\relV{R}{\xid}{\version_{n+1}}$ and does not occur in $\relCV{R}{\version_{n+1}}$, function $\isMax($ $\relCV{R}{\version_{n+1}},k)$  returns $1$ and $k$ is in $\vMerge(\ract(u_{n}),\relCV{R}{\version(u_{n+1})})(t)$. Similarly, since $k$ is the latest version, we have that 
$\validIn($ $\relV{R}{\xid}{\version_{n+1}},t,k, \version_{n+1})$ returns $1$ because $k$ has a version annotation from $\xid$ as the outmost version annotation. Thus, $k$ is also present in $\relVE{R}{\xid}{\version_{n+1}}$.

\prooftitle{Case 2}
Summand $k$ is the only summand with identifier $\tid$ in $\relV{R}{\xid}{\version_{n+1}}$. While there exists a summand $k'$ with identifier $\tid$ in the annotation of some tuple $t'$ in $\relCV{R}{\version_{n+1}}$, we know that $\version(k') < \version(k)$. Thus, function $\isMax(\relCV{R}{\version_{n+1}},k)$  returns $1$ and $k$ is in $\vMerge(\ract(u_{n}),\relCV{R}{\version(u_{n+1})})(t)$. Function $\validIn(\relV{R}{\xid}{\version_{n+1}})$ returns $1$ for the same reason as in case 1 above.

Now consider summand $k'$ with $\idOf(k') = \tid$ that occurs as a summand in the annotation of tuple $t'$ in $\relCV{R}{\version_{n+1}}$. We have to show that both $\isStrictMax$ and $\validEx$ return $0$ for this outdated tuple version. $\isStrictMax(\relV{R}{\xid}{\version_{n+1}}, k') = 0$, because the summand $k$ occurs in $\relV{R}{\xid}{\version_{n+1}}(t)$,  $\idOf(k) = \idOf(k')$, and $\versionOf(k) > \versionOf(k')$. Also $\validEx($ $\relCV{R}{\version_{n+1}},t',k', \version_{n+1})$ returns $0$, because $\hasUp(\xid, t',k',\version_{n+1})$ evaluates to true.

\prooftitle{Case 3}
Since $k$ is the only summand with identifier $\tid$ in $\relCV{R}{\version_{n+1}}$ and does not occur in $\relV{R}{\xid}{\version_{n+1}}$, function $\isStrictMax$ $(\relV{R}{\xid}{\version_{n+1}},k)$  returns $1$ and $k$ is in $\vMerge(\ract(u_{n}),\relCV{R}{\version(u_{n+1})})(t)$. Similarly, since $k$ is the latest version of a tuple version with identifier $\idOf(k)$, we have that 
$\validEx($ $\relCV{R}{\version_{n+1}},t,k, \version_{n+1})$ returns $1$ because $\hasUp(\xid, t,k,\version_{n+1})$ evaluates to false. Thus, $k$ is also present in $\relVE{R}{\xid}{\version_{n+1}}$.

\prooftitle{Case 4}
Summand $k$ is the only summand with identifier $\tid$ in $\relCV{R}{\version_{n+1}}$. While there exists a summand $k'$ with identifier $\tid$ in the annotation of some tuple $t'$ in $\relV{R}{\xid}{\version_{n+1}}$, we know that $\version(k') < \version(k)$. Thus, function $\isStrictMax(\relV{R}{\xid}{\version_{n+1}},k)$  returns $1$ and $k$ is in $\vMerge(\ract(u_{n}),\relCV{R}{\version(u_{n+1})})(t)$. Function $\validEx(\relCV{R}{\version_{n+1}})$ returns $1$ for the same reason as in case 3 above.

Now consider summand $k'$ with $\idOf(k') = \tid$ that occurs as a summand in the annotation of tuple $t'$ in $\relV{R}{\xid}{\version_{n+1}}$. We have to show that both $\isMax$ and $\validIn$ return $0$ for this outdated tuple version. $\isMax(\relCV{R}{\version_{n+1}}, k') = 0$, because there is summand $k$ in $\relCV{R}{\version_{n+1}}(t)$, $\idOf(k) = \idOf(k')$, and $\versionOf(k) > \versionOf(k')$. Also $\validIn(\relV{R}{\xid}{\version_{n+1}},$ $t',k', \version_{n+1})$ returns $0$, because $k$ does not have a version annotation from $\xid$ as its outermost version annotation.

\prooftitle{Case 5}
Summand $k$ was created by a Transaction $\xid'$ with $\finish{\xid'} < \start{\xid}$. Thus, $k$ is present in both $\relV{R}{\xid}{\version_{n+1}}(t)$ and  $\relCV{R}{\version_{n+1}}(t)$ and based on the definition of this case no other summand $k'$ with $\idOf(k) = \idOf(k')$ occurs in the annotation of any tuple $t'$ in $\relV{R}{\xid}{\version_{n+1}}(t)$ or  $\relCV{R}{\version_{n+1}}(t)$. Thus, $\isMax(\relCV{R}{\version_{n+1}},k)$ returns 1 because there is no newer version of $k$ in $\relCV{R}{\version_{n+1}}$ while $\isStrictMax(\relV{R}{\xid}{\version_{n+1}}, k)$ returns  $0$, because there exists $k$ in $\relV{R}{\xid}{\version_{n+1}}$. Similarly, $\validIn$ returns $0$ because $\xid$ has not created tuple version $k$ whereas $\validEx$ evaluates to $1$, because $\xid$ has not updated $k$. 

Having proven all cases we get:

\begin{align*}
& \sum_{i=0}^{\numInSum{\relV{R}{\xid}{\version_{n+1}}(t)}} \relV{R}{\xid}{\version_{n+1}}(t)[i] \\
&\mathbigtab\mathbigtab\times \isMax(\relCV{R}{\version_{n+1}},\relV{R}{\xid}{\version_{n+1}}(t)[i]) \\ 
&+\sum_{i=0}^{\numInSum{\relCV{R}{\version_{n+1}}(t)}} \relCV{R}{\version_{n+1}}(t)[i] \\
&\mathbigtab\mathbigtab\times \isStrictMax(\relV{R}{\xid}{\version_{n+1}},\relCV{R}{\version_{n+1}}(t)[i])\\
  =&\sum_{i=0}^{\numInSum{\relV{\rel}{\xid}{\version_{n+1}}(t)}} \hspace{-4mm} \nthOfK{\relV{\rel}{\xid}{\version_{n+1}}(t)}{i}\\
&\mathbigtab\mathbigtab \times \validIn (\xid,t, \nthOfK{\relV{\rel}{\xid}{\version_{n+1}}(t)}{i}, \version_{n+1})\\
&+ \hspace{-4mm}
\sum_{i=0}^{\numInSum{\relCV{\rel}{\version_{n+1}}(t)}} \hspace{-2mm} \nthOfK{\relCV{\rel}{\version_{n+1}}(t)}{i} \\
&\mathbigtab\mathbigtab\times \validEx (\xid,t, \nthOfK{\relCV{\rel}{\version_{n+1}}(t)}{i}, \version_{n+1})\\[2mm]
\end{align*}
Reordering the two sums we get
\begin{align*}
=&\sum_{i=0}^{\numInSum{\relCV{\rel}{\version_{n+1}}(t)}} \hspace{-2mm} \nthOfK{\relCV{\rel}{\version_{n+1}}(t)}{i} \\
&\mathbigtab\mathbigtab\times \validEx (\xid,t, \nthOfK{\relCV{\rel}{\version_{n+1}}(t)}{i}, \version_{n+1})\\
&+ \hspace{-4mm}
\sum_{i=0}^{\numInSum{\relV{\rel}{\xid}{\version_{n+1}}(t)}} \hspace{-4mm} \nthOfK{\relV{\rel}{\xid}{\version_{n+1}}(t)}{i} \\
&\mathbigtab\mathbigtab\times \validIn (\xid,t, \nthOfK{\relV{\rel}{\xid}{\version_{n+1}}(t)}{i}, \version_{n+1})\\
= &\relVE{R}{\xid}{\version_{n+1}}(t)\\
=&\relV{R}{\xid}{\version_{n+1}}(t) 
\end{align*}
Thus, we have shown that $\xid \equiv_{\ppSRV} \ract(\xid)$.
\\[3mm]

\prooftitle{$\xid \equiv_{\ppSRV} \ract_{opt}(\xid)$}

Let $\xid = u_1, \ldots, u_n, c$ be a transaction in a RC-SI history $\history$. Recall that $\ract_{opt}$ is evaluated over $\relCV{\rel}{\finish{\xid} - 1}$. As shown in Lemma~\ref{lem:end-minus-one-has-all-pred} for any given tuple identifier $\tid$, $\relCV{\rel}{\finish{\xid} - 1}$ contains the predecessor of the earliest version of a tuple with identifier $\tid$ created by Transaction $\xid$ (if such a tuple version exists). Thus, the reenactment is correct as long as the following three conditions hold: 1) the first update in $\xid$ that creates a new version of with identifier $\tid$ updates this version in $\ract_{opt}$; 2) the reenactment query for each update $u_i \in \xid$ is not applied to any tuple version $k$ from $\relCV{\rel}{\finish{\xid} - 1}$ with $\version(k) > \version(u_i)$; and 3) each tuple version $k$ is passed on by the reenactment query for each $u_j$ with  $\version(k) > \version(u_j)$.

We prove this by induction over the position of an update in $\xid= u_1, \ldots, u_n, c$.

\prooftitle{Induction Start}
Consider $u_1$, the first update of $\xid$. Update $u_1$ is either an update, delete, or simple insert (the insert's query is a singleton operator $\asingleton{t}{k}$). Let $\version_1$ denote $\version(u_1)$ and $\version_e$ to denote $\finish{\xid} - 1$.

\smallskip\noindent\textit{$u_1$ is an update:} First consider the case where $u_1$ is an update. The part of the reenactment query for $\xid$ corresponding to $u_1$ is
\begin{align*}
&\annotOp{U}{\xid}{\version_1+1}{}(\projection_{A}(\selection_\theta(\vFilt{V \leq \version_1}(\relCV{R}{\version_e}))))\\ 
                                                  \union &\selection_{\neg \theta}(\vFilt{V \leq \version_1}(\relCV{R}{\version_e})) \\
\union &\vFilt{V > \version_1}(\relCV{R}{\version_e}))
\end{align*}

Based on Lemma~\ref{lem:end-minus-one-has-all-pred}, $\relCV{\rel}{\version_e}$ contains all versions of tuples that got updated by $u_1$. 
Consider a tuple version $k$ in the annotation of a tuple $t$ in $\relCV{\rel}{\version_e}$. Depending on whether $\version(k) \leq \version_1$ holds or not, this tuple version will be visible to $u_1$ or not. If $k$ is visible to $u_1$ then whether $k$ will be updated  depends on whether $t$ fulfills the update's condition or not. If $\version(k) > \version_1$ then $k$ will not fulfill the condition $V \leq \version_1$ of the version filter operators in the first two branches of the union. Tuple version $k$ fulfills the condition of the third branch ($V > \version_1$) and, thus, will be passed on unmodified to the output of the part of the reenactment query corresponding to $u_1$. This implies that the second and third correctness conditions introduced above hold (non-visible tuple versions are not updated and passed on unmodified). If $k$ was visible to $u_1$ and was updated by $u_1$, then we know that $\version(k) \leq \version_1$. Thus, $k$ fulfills the condition $V \leq \version_1$ of the version filter operator in the first two branches of the union, but only fulfills the selection condition ($\theta$) of the first branch of the union and, thus, is updated (first condition). Note that if $k$ was visible to $u_1$, but was not updated by $u_1$ then either $k$ will be ``routed'' through the second branch of the union (if $k$ is the latest version of a tuple with identifier $\idOf(k)$ present in $\relCV{\rel}{\version_e}$) or $k$ will not be in $\relCV{\rel}{\version_e}$ (if $\relCV{\rel}{\version_e}$ contains a newer version of a tuple with identifier $\idOf(k)$).

\smallskip\noindent\textit{$u_1$ is a delete:} The part of the reenactment query for $\xid$ corresponding to a delete $u_1$ is

\begin{align*}
&\annotOp{D}{\xid}{\version_1+1}{}(\selection_\theta(\vFilt{V \leq \version_1}(\relCV{\rel}{\version_e})))\\ 
\union &\selection_{\neg \theta}(\vFilt{V \leq \version_1}(\relCV{\rel}{\version_e}))\\
\union &\vFilt{V > \version_1}(\relCV{\rel}{\version_e})
\end{align*}

Consider a tuple version $k$ in the annotation of a tuple $t$ in $\relCV{\rel}{\version_e}$. Note that the third branch of the union is identical for updates and deletes. Hence, if $\version(k) \leq \version_1$, the second and third conditions hold. The cases where $k$ is affected by $u_1$ or $k$ is visible, but not affected, are also analog to the proof for updates. 

\smallskip\noindent\textit{$u_1$ is a simple insert:} The part of the reenactment query for $\xid$ corresponding to a delete $u_1$ is

\begin{align*}
\annotOp{I}{\xid}{\version_1+1}{}(\asingleton{t}{k}) \union \relCV{\rel}{\version_e}
\end{align*}

All tuples from $\relCV{\rel}{\version_e}$ are present in the output (second and third condition) and, since an insert creates new tuple versions, the first condition trivially holds.

\prooftitle{Induction Step} 

We have to show that under the assumption that updates $u_j$ with $j \leq i$ are reenacted correctly by $\ract_{opt}$, then the same holds for $u_{i+1}$. Let $\version_{i+1}$ denote $\version(u_{i+1})$. Again this has to be shown for the three cases of $u_{i+1}$ being an 1) update, 2) delete, or 3) simple insert. Observe that the input to the part of the reenactment query corresponding to $u_j$ is equal to $\relCV{\xid}{\version_e}$ except that some tuple versions have been replaced by updated tuple versions by the part of the reenactment query corresponding to updates $u_1$ to $u_i$. Since this is the only difference to the induction start, we only have to prove this additional case. Consider such a version $k$ of tuple $t$ produced by $u_j$ with $j \leq i$. It follows that $\version(k) \leq \version_{i+1}$. Thus, $k$ fulfills the conditions of the first two branches of the union for updates and deletes. Based on the induction hypothesis, if $u_{i+1}$ produces a tuple with identifier $\idOf(k)$ then $k$ is the previous version of this tuple. Thus, the update's respective deletion's condition $\theta$ evaluates to true for $t$ and $k$ will be updated respective deleted. If $k$ does not fulfill the condition then the second branch of the union passes on $k$ unmodified. It follows that $u_{i+1}$ is correctly reenacted by $\ract_{opt}(\xid)$.

\end{proof}

Note that based on the results of~\cite{AG16} equivalence under $\ppSRV$ implies equivalence under any naturally ordered MV-semiring $\mvK$. Furthermore, it was proven~\cite{AG16} that if  $\semK$ is naturally ordered, then so is $\mvK$.  The first result follows from commutation of queries and transactional histories with lifted homomorphisms. Based on Lemmas~\ref{lem:lift-h-commutes-RC-SI-histories},~\ref{lem:lifted-commutes-version-merge} and ~\ref{lem:lifted-commutes-version-filter} such homomorphisms also commute with the new query operators we have introduced and RC-SI histories. Thus, $\ppSRV$ implies equivalence under any naturally ordered MV-semiring $\mvK$ for any of the operations used in this paper. 

\end{document}